\documentclass[reprint, onecolumn, showkeys, aps, pra,superscriptaddress]{revtex4-2}
\usepackage{capt-of}
\usepackage[]{graphicx, floatrow}
\newfloatcommand{capbtabbox}{table}[][\FBwidth]
\usepackage[T1]{fontenc}
\usepackage{calligra}

\usepackage[nolist]{acronym}
\usepackage{amsthm}
\newtheorem{theorem}{Theorem}

\newtheorem{lemma}[theorem]{Lemma}

\usepackage{subcaption}
\usepackage{tabularx}
\usepackage{booktabs}
\usepackage{fancyhdr}
\usepackage{ulem}
\usepackage{soul}
\usepackage{color}
\usepackage{graphicx}
\usepackage{enumitem}
\usepackage{amsfonts}
\usepackage{amssymb}
\usepackage{amsmath,latexsym,amsthm}
\usepackage{physics}
\usepackage{bm}
\usepackage{diagbox}
\usepackage{multirow}
\usepackage[bottom]{footmisc}
\usepackage[breaklinks=true,colorlinks=true,linkcolor=teal,urlcolor=teal,citecolor=teal]{hyperref}
\usepackage{pgfplots}
\pgfplotsset{width=10cm,compat=1.9}
\usepackage{tikz}
\usetikzlibrary{matrix,decorations.pathreplacing,calc,positioning,calligraphy}
\tikzstyle{block} = [rectangle, draw, fill=white, rounded corners,
                 minimum height=6em, minimum width =8em, text width=7em, text centered]
\tikzstyle{midblock} = [rectangle, draw, fill=white, rounded corners,
                 minimum height=6em, minimum width =11em, text width=10em, text centered]
\tikzstyle{wideblock} = [rectangle, draw, fill=white, rounded corners,
                 minimum height=6em, minimum width =13em, text width=12em, text centered]
\tikzstyle{xwideblock} = [rectangle, draw, fill=white, rounded corners,
                 minimum height=6em, minimum width =22em, text width=21em, text centered]
\tikzstyle{xxwideblock} = [rectangle, draw, fill=white, rounded corners,
                 minimum height=6em, minimum width =27em, text width=26em, text centered]
\tikzstyle{xxxwideblock} = [rectangle, draw, fill=white, rounded corners,
                 minimum height=6em, minimum width =40em, text width=37em, text centered]

\usetikzlibrary{quantikz2}
\usepgfplotslibrary{groupplots} 
\usepackage[dvipsnames]{xcolor}
\colorlet{OAA}{blue}
\colorlet{QFT}{Salmon}
\colorlet{SP}{violet}
\colorlet{FABLE}{orange}

\usepackage[nameinlink, capitalise]{cleveref}

\AddToHook{cmd/appendix/before}{\crefalias{section}{appendix}}

\begin{document}

\title{Quantum algorithms for solving a drift-diffusion equation: analysing circuit depths}

\author{Ellen Devereux}
\email{Ellen.Devereux@warwick.ac.uk}
\affiliation{Department of Physics, University of Warwick, Coventry, CV4 7AL, United Kingdom}
\affiliation{Fujitsu UK Ltd}


\author{Animesh Datta}
\affiliation{Department of Physics, University of Warwick, Coventry, CV4 7AL, United Kingdom}

\date{\today}

\begin{abstract}
We compare the circuit depths for five different gate sets to implement a quantum algorithm solving a \acl{DDE} in two spatial dimensions. Our algorithm uses diagonalisation by the \acl{QFT}. The gate sets are: An unconstrained gate set, the TK1 gate set from Quantinuum, the native gate sets of IBM Heron and IonQ, and Fujitsu's \ac{STAR} gate set. 
Our analysis covers a set of illustrative scenarios using up to 22 qubits.
We find that while scaling with spatial resolution aligns with theoretical predictions, scaling with spatial dimension is less efficient than theorised due to overhead from block encoding. Finally, using the \ac{STAR} gate set, we find that even minimal problem instances exceed the operational limits of current quantum hardware.
\end{abstract}

\keywords{Quantum Computing, Modelling, Partial Differential Equations}

\maketitle

\begin{acronym}
\acro{DDE}[DDE]{drift-diffusion equation}
\acro{QFT}[QFT]{quantum Fourier transform}
\acro{IQFT}[IQFT]{inverse \ac{QFT}}
\acro{FFT}[FFT]{fast Fourier transform}
\acro{IFFT}[IFFT]{inverse \ac{FFT}}
\acro{FTCS}[FTCS]{forward time centered space}
\acro{FABLE}{fast approximate block encoding}
\acro{OAA}{oblivious amplitude amplification}
\acro{STAR}{space-time efficient analog rotation}
\acro{NISQ}{noisy intermediate scale quantum}
\acro{PDE}{partial differential equation}
\end{acronym}
\section{Introduction}

Quantum computing holds promise for tackling computationally intensive problems~\cite{Dalzell2023QuantumComplexities}. Yet significant challenges remain in bridging the gap between theorised capabilities and practical hardware implementations. Resource analyses have been made for various algorithms, such as Shor's algorithm~\cite{Gheorghiu2025QuantumAlgorithms, Gidney2021HowQubits, Yamaguchi2023ExperimentsSimulator}, which promise exponential computational advantage on fault-tolerant quantum computers. Other studies focus on near-term applications for \ac{NISQ} computers~\cite{Katabarwa2024EarlyComputing, Kiss2025EarlyEstimation}.
The latter take into account the resources required for error correction of near-term devices.

To bridge the gap between the theoretical promises and the current realities of quantum computing, this paper performs a practical analysis of the total number of operations required to solve the \ac{DDE}.
The \ac{DDE} is a fundamental mathematical model with broad applicability, ranging from modelling wind turbine power outputs~\cite{Arenas-Lopez2020AOutput} to pricing financial options~\cite{Montanaro2015QuantumMethods, Andersen2001TheVolatility, Wang2023ModelingProcess}. Its widespread use underscores the importance of efficient solution algorithms. Quantum computing offers a theoretical computational advantage for solving the DDE~\cite{Novikau2024QuantumSystems, Lubasch2025QuantumSpace, Brearley2024QuantumSimulation, Devereux2025QuantumEquation, NovikauExplicitEquation, Over2025QuantumOperator}, making it a compelling candidate for quantum acceleration.

We focus on the most efficient quantum algorithm known -- the diagonalisation via \ac{QFT}~\cite{Devereux2025QuantumEquation} and analyse its requirements against current device capabilities. 
Our analysis accounts for practical circuit implementation using theoretical as well as native universal gate sets. We define a native gate set to be one implemented on a real-world quantum hardware. We consider five gate sets: an unconstrained gate set that allows any possible one, two or three qubit gate, the tket theoretical gate set from Quantinuum~\cite{Sivarajah2020Tket:Devices}, the IBM Heron native gate set~\cite{ProcessorDocumentation, Montanez-Barrera2025EvaluatingDepth}, Fujitsu's \acf{STAR} native gate set~\cite{Toshio2025PracticalComputer} and the IonQ native gate set~\cite{NativeDocumentation}. These gate sets are detailed in~\cref{tab: gate sets}.
We note that previous work benchmarking quantum computers found the IBM Heron class of quantum computers to be the highest performing in terms of number of coherent operations~\cite{Montanez-Barrera2025EvaluatingDepth}.
We quantify the circuit depth 
for a set of illustrative scenarios up to 22 qubits. A summary of the scenarios studied herein 
is in~\cref{tab:scenarios}. 

We distil four main results: 
(i) The circuit depth scales as expected with the number of spatial discretisation steps in one dimension. This confirms one of the primary findings of the complexity analysis of the best-known quantum algorithm for solving the \ac{DDE} - diagonalisation via quantum Fourier transform
(QFT)~\cite{Devereux2025QuantumEquation}.
(ii) While theoretical gate assemblies (such as tket~\cite{Sivarajah2020Tket:Devices}) are most efficient in circuit depth, the \ac{STAR}~\cite{Toshio2025PracticalComputer} gate set is the most efficient \textit{native} gate set we tested.
(iii) Circuit depth scales far more steeply with spatial dimension than theorised due to the overhead of the \ac{FABLE}~\cite{Camps2022FABLE:Block-Encodings} subroutine for block encoding. This discrepancy underscores the gap between theoretical complexity and practical quantum implementation. (iv) Even for minimal problem instances, the required circuit depths 
exceed the capabilities of current quantum computing hardware.
A complete resource estimation for fault-tolerant implementations is reserved for future work, when error-corrected circuit generators become more readily available.

The remainder of this paper is structured as follows: \cref{ssec: problem} introduces the \ac{DDE} problem statement and outlines the solution algorithms, including an analytic solution, the classical diagonalisation algorithm, and the quantum diagonalisation algorithm~\cite{Devereux2025QuantumEquation}  
\cref{sec: build} details the construction of the quantum circuit using current best practices. Circuit simulation and validation against analytic and classical solutions are presented in \cref{sec: results}. In \cref{Sec: res_est}, we provide a comprehensive analysis of circuit depth across five different gate sets, with a deeper subroutine analysis focusing on the \ac{STAR} gate set~\cite{Toshio2025PracticalComputer} and the tket gate set~\cite{Sivarajah2020Tket:Devices}. Finally, in~\cref{sec:discuss} we compare our findings to the capabilities of current quantum computers and discuss the outcomes. 

\subsection{Problem Statement}
\label{ssec: problem}
The \acl{DDE} we solve in this paper is defined as 
\begin{equation}
    \frac{\partial p(\mathbf{x}, t)}{\partial t}
    = \sum_{i = 1}^d \bigg[ a \frac{\partial}{\partial x_i} [p(\mathbf{x}, t)] + D \frac{\partial^2}{\partial x_i^2} [p(\mathbf{x}, t)]\bigg],
    \label{eqn: DDE}
\end{equation}
where $\mathbf{x} = \{x_1,\cdots,x_d\}~\in \mathbb{R}^d$ is a vector and $a$ and $D$ are positive constants representing the drift and diffusion coefficients, respectively. 
We seek an approximate solution, $\Tilde{\Tilde{p}}(\mathbf{x},t)$, of $p(\mathbf{x},t)$ from \cref{eqn: DDE} up to a given error $\epsilon \in (0,1)$ at a time $t = T$ by
\begin{equation}
    \big|\big|\Tilde{\Tilde{p}}(\mathbf{x}, t) - p(\mathbf{x}, t) \big| \big|_{\infty}\leq \epsilon  
    \label{eqn: approx_definition}
\end{equation}
for $\mathbf{x} \in [-L,L]^d$. The infinity norm is defined as $||\mathbf{x}||_{\infty}=\textrm{max}_j|x_j|$.
The solution $p(\mathbf{x}, t)$ is positive   
and $\int_{-L}^L p(\mathbf{x}, t) d\mathbf{x} = 1$ at all times \cite{Risken1996Fokker-PlanckEquation}. $p(\mathbf{x}, t)$ is also dimensionless.
We assume periodic boundary conditions in all spatial dimensions $x_j$ but not in $t$ such that $p(L, t) = p(- L, t)$. We also require that $p(\mathbf{x}, t)$ is four times differentiable~\cite{Devereux2025QuantumEquation}.

We discretise \cref{eqn: DDE} on a grid referred to as $G$, which spans $[-L, L]^d$ in space and $[0, T]$ in time.
This grid is discretised into $n_x$ equally spaced points in each of the $d$ spatial dimensions ($n_x$ must be even) and $n_t$ equally distributed points in time. 
The spacing of these points is $\Delta x = 2L/n_x$ in all space dimensions and $\Delta t = T/n_t$.

We use the \ac{FTCS} discretisation strategy to produce the approximate solution 
\begin{align}
        \tilde{p}(\mathbf{x}, t + \Delta t)  &= \mathcal{L}\tilde{p}(\mathbf{x}, t)  \nonumber\\ \nonumber
      &=  	\bigg(1 - \frac{2d D\Delta t}{\Delta x^2}\bigg) \tilde{p}(\mathbf{x}, t) + \Delta t \sum_{j = 1}^d \bigg ( \frac{a }{2\Delta x}(\tilde{p}(...,x_{j[i]} + \Delta x,...,t) - \tilde{p}(...,x_{j[i]} - \Delta x,...,t))
        \\&+ \frac{D}{\Delta x^2}(\tilde{p}(...,x_{j[i]} + \Delta x,...,t)
    + \tilde{p}(...,x_{j[i]} - \Delta x,...,t))
        \bigg)
        \label{eqn: diff operator}
    \end{align}
to define $\mathcal{L}$ as a linear operator of dimension $n_x^d \times n_x^d$. The error introduced by this discretisation is denoted as $\epsilon_c$,
\begin{equation}
    ||\Tilde{p}(\mathbf{x}, t) - p(\mathbf{x}, t)||_{\infty} \leq \epsilon_c.
    \label{eqn: def_epsc}
\end{equation}
For the quantum algorithm, we use Dirac notation such that the approximate probability distribution $\tilde{p}(\mathbf{x}, t)$ is encoded in the quantum state 
\begin{equation}
        |\tilde{p}\rangle = \frac{1}{||\tilde{p}(\mathbf{x}, t)||_{2}}\sum_{(\mathbf{x},t)\in G} \tilde{p}(\mathbf{x},t)|\mathbf{x},t\rangle.
        \label{eqn: quantum_state}
    \end{equation}   
We can then encode this solution using $q = d\log n_x$ qubits. As we can only measure the above quantum state to a finite error,  the extracted approximated probability distribution $\tilde{\tilde{p}}(\mathbf{x}, t)$ is such that
     \begin{equation}
        ||\tilde{\tilde{p}}(\mathbf{x}, t) - \tilde{p}(\mathbf{x}, t)||_{\infty}\leq \epsilon_q
        \label{eqn: def_epsq}
    \end{equation}
The error in approximating the quantum state is thus $\epsilon_q,$ and the overall error is $\epsilon = \epsilon_c + \epsilon_q$ as in \cref{eqn: approx_definition}.  
We will use the most efficient quantum algorithm for the solution identified by us in Ref.~\cite{Devereux2025QuantumEquation} to solve the \ac{DDE}. We restate it here for completeness.

\begin{theorem}[Quantum diagonalisation {\cite[Theorem 21] {Devereux2025QuantumEquation}}] There is a quantum algorithm that estimates the probability distribution $p(\mathbf{x})$ from Eqn.~\eqref{eqn: DDE} such that  $||\Tilde{\Tilde{p}}(\mathbf{x}, T) - p(\mathbf{x}, T)||_{\infty} \leq \epsilon_c +\epsilon_q $ for all $(\mathbf{x}) \in G$ at a fixed $T=n_t\Delta t$ with $99\% $ success probability in time 
   \begin{equation}
        \Tilde{O}\bigg(\frac{d^2n_x^{d/2}}{\epsilon_q} \bigg) = \Tilde{O}\bigg( \frac{d^{(d/2+2)}T^{d/2}\zeta^{d/4} (aL + D)^{d/2}}{ \epsilon_q\epsilon_c^{d/4}}\bigg).
        \label{eqn: qft_comp}
    \end{equation}
    \label{thm: QFT}
\end{theorem}
Our algorithm requires $q = d\log_2 n_x$ qubits to encode the system and assumes one ancilla qubit for amplitude amplification. Therefore, the theorised circuit width, $w$, for this application is $w = q+1 = d\log n_x +1$. We also note that $n_t \propto n_x^2$~\cite[Corollary 2]{Devereux2025QuantumEquation}.
Our algorithm comprises of six stages, also depicted in~\cref{fig: AA_diag}:
\begin{enumerate}
    \item State Preparation: Load the initial condition into the quantum computer.
    \item Diagonalisation: Diagonalise the problem using the \ac{QFT}.
    \item Eigenvalue Application: Apply the eigenvalues of the operator $\mathcal{L}$ using a diagonal matrix $\Lambda$.
    \item Amplitude Amplification: Amplify the amplitude of the result to increase the probability of success.
    \item Inverse QFT: Apply the \ac{IQFT}.
    \item Measurement: Measure the resulting state to obtain the solution.
\end{enumerate}

Throughout this paper, we will use the analytic solution in one dimension and the classical diagonalisation algorithm~\cite[Theorem 10]{Devereux2025QuantumEquation} for validation of our quantum algorithm.
The analytic solution of~\cref{eqn: DDE} in one spatial dimension is~\cite{Risken1996Fokker-PlanckEquation}
\begin{equation}
        p(x,t|x_0,t_0) =  \dfrac{1}{\sqrt{4\pi D(t-t_0) + 2\pi}} e^{-\dfrac{(a(t-t_0)+(x-x_0))^2}{4D(t-t_0)+2}}.
    \label{eqn: analytic}
\end{equation}
For validation in higher dimensions, we use the classical diagonalisation by \ac{FFT}~\cite[Theorem 10]{Devereux2025QuantumEquation}. This follows the same process as the quantum algorithm but uses the \ac{FFT} and \ac{IFFT} instead of the \ac{QFT}. 

\begin{theorem}[Classical diagonalisation by fast Fourier transform {\cite[Theorem 10]{Devereux2025QuantumEquation}}] There is a classical algorithm that outputs an approximate solution $\tilde{p}(\mathbf{x}, t)$ such that $||\tilde{p}(\mathbf{x}, t) - p(\mathbf{x}, t)||_{\infty} \leq \epsilon_c$ for all $(\mathbf{x}, T)\in G$ at final time $T$, in time
    \begin{equation}
        \tilde{O}(dn_x^d) = \tilde{O}\bigg(\frac{d^{d/2+1}(T\zeta)^{d/2}L^d(aL + D)^d}{\epsilon_c^{d/2} D^{d/2}}\bigg).
    \end{equation}
    \label{thm: Classical_FFT}
\end{theorem}

Amplitude amplification is not required in the classical case.

In the following section, we build the quantum circuit for our algorithm in \cref{thm: QFT}. We introduce each of the subroutines and comment on their expected scaling.

\label{sec: intro}

\section{Building the circuit}

The quantum diagonalisation by \ac{QFT} algorithm is made up of six stages, as enumerated after~\cref{thm: QFT}.
Our resultant quantum circuit is drawn in \cref{fig: AA_diag}.
We now discuss the available subroutines for each of these steps.

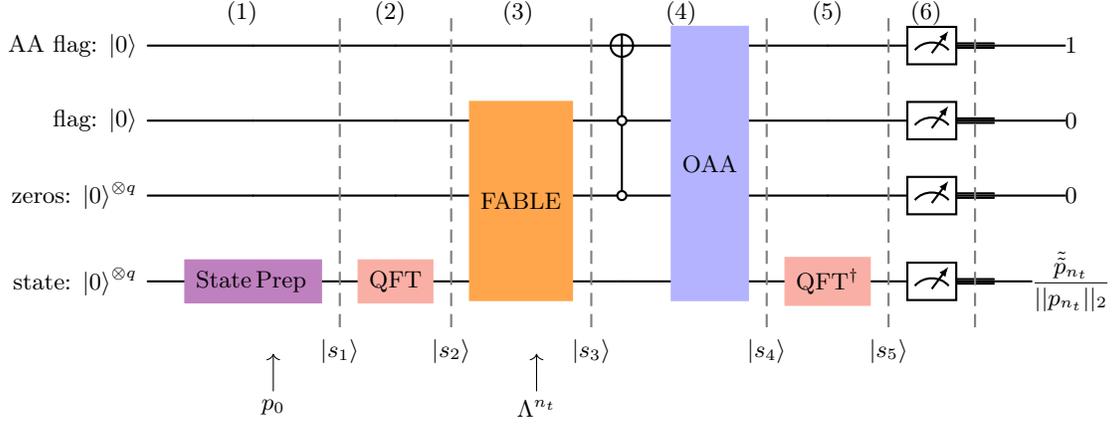
\begin{figure}
    \centering
\begin{tikzpicture}
 \node (a) at (0,0)
 {
\begin{quantikz}
 \lstick{{AA flag: $\ket{0}$}} & & & & \targ{} & \gate[wires=4, style={OAA!30}]{\text{OAA}}\slice[style=gray, label style={label={[xshift=-3.5em, yshift=12em, black]above:(4)}, pos=1, anchor = north}]{$\ket{s_4}$} & & \meter{} & \cw & 1 \\
 \lstick{{flag: $\ket{0}$}} & & & \gate[wires=3, style={FABLE!70}]{\mathrm{FABLE}}\slice[ style=gray, label style={label={[xshift=-3em, yshift=12em, black]above:(3)}, pos=1, anchor = north}]{$\ket{s_3}$} & \ctrl[open]{-1}& \hphantom{\mathrm{OAA}} & & \meter{} & \cw & 0 \\
 \lstick{{zeros: $\ket{0}^{\otimes q}$}} & & & \hphantom{\mathrm{FABLE} } & \ctrl[open]{-2}& \hphantom{\mathrm{OAA}} & & \meter{} & \cw & 0 \\
 \lstick{{state: $\ket{0}^{\otimes q}$} } & \gate[style = {SP!50}]{ \mathrm{State \,Prep} \, }\slice[style=gray, label style={label={[xshift=-4em, yshift=12em, black]above:(1)}, pos=1, anchor = north}]{$\ket{s_1}$} & \gate[style={QFT!70}]{ \mathrm{QFT}}\slice[style=gray, label style={label={[xshift=-2.5em, yshift=12em, black]above:(2)}, pos=1, anchor = north}]{$\ket{s_2}$} & \hphantom{\mathrm{FABLE}} & & \hphantom{\mathrm{OAA}}& \gate[style={QFT!70}]{ \mathrm{QFT}^{\dag}}\slice[style=gray, label style={label={[xshift=-2.5em, yshift=12em, black]above:(5)}, pos=1, anchor = north}]{$\ket{s_5}$} & \meter{}\slice[style=gray, label style={label={[xshift=-2 em, yshift=12em, black]above:(6)}, pos=1, anchor = north}]{} & \cw & \dfrac{\tilde{\tilde{p}}_{n_t}}{||p_{n_t}||_2}
\end{quantikz}
 };
 \node (b) at (a.south) [anchor=north,xshift=-2cm, yshift=0.5cm]
 {
 \begin{tikzpicture}[y=-1cm]
 \draw[<-] (1,2.5) -- (1,3) node[below] {$p_0$}; 
 \draw[<-] (4.5,2.5) -- (4.5,3) node[below] {$\Lambda^{n_t}$}; 
 \end{tikzpicture}
 };

\end{tikzpicture}
    \caption{The quantum circuit for solving the \ac{DDE} using the diagonalisation by \ac{QFT} method from~\cref{thm: QFT}. 
    Comprised of the six stages noted in~\cref{sec: intro}. The intermediate states $\ket{s_i}$ are described as follows: $\ket{s_1} = \ket{0}\ket{0}\ket{0}^{\otimes q}\ket{p_0}$, $\ket{s_2} = \ket{0}\ket{0}\ket{0}^{\otimes q}\ket{\psi_0}$,  where $\ket{\psi_0}$ is defined in~\cref{eq: psi0}, 
    and $\ket{s_3}$, $\ket{s_4}$ and $\ket{s_5}$ are defined in~\cref{eq:s3,eq:s4,eq:s5} respectively.
    The state preparation and the \acs{FABLE} subroutines encode classically computed values denoted with arrows. The width of the circuit is $w=2q +2 = 2 d\log n_x + 2$.}
    \label{fig: AA_diag}
\end{figure}

(1) 
The first stage of the quantum circuit is state preparation. This step builds the state $\ket{p_0}$ on the state register of $q$ qubits. Ref. \cite{Araujo2023QuantumLibrary} carries out an analysis of state preparation algorithms. We chose to implement the most efficient of these algorithms, the low-rank state preparation subroutine~\cite{Araujo2024Low-RankPreparation}. 
An analysis of its circuit depth, compared to built-in algorithms from tket~\cite{Sivarajah2020Tket:Devices} and the theoretical complexity, is provided in \cref{app: SP}.
The contribution of the state preparation subroutine to the overall circuit and analysis will be depicted in~\cref{fig:gate sets d=1and2} using \textcolor{SP}{mauve pentagons}.

(2) 
The second stage of the circuit is the \ac{QFT}. This transforms the initial condition $\ket{p_0}$ to the state 
\begin{equation}
    \ket{\psi_0} = \mathcal{F}\ket{p_0}.
    \label{eq: psi0}
\end{equation}
The complexity of the \ac{QFT} scales as $O(q^2)$~\cite{Nielson2010QuantumInformation, PfefferMultidimensionalTransformation, Musk2020AComputations}. Although a reduced circuit depth is possible using an approximate version~\cite{Nam2020ApproximateGates}, we utilise the exact \ac{QFT} as it is not a limiting factor in our algorithm. 
The contribution of the \ac{QFT} subroutine to the overall circuit and analysis will be depicted in~\cref{fig:gate sets d=1and2} using \textcolor{QFT}{pink circles}.

(3) 
The third step in our quantum circuit is to efficiently map the eigenvalues of the operator $\mathcal{L}$ onto the state $\ket{\psi_0}$. We do this using the operator $\Lambda$, which is a diagonal matrix of the eigenvalues of $\mathcal{L}$.
The eigenvalues are ~\cite[Eqn. 32]{Devereux2025QuantumEquation}
\begin{equation}
        l_j = 1 - \frac{4Dd\Delta t}{\Delta x^2}\sin^2\bigg(\frac{\pi j}{n_x}\bigg) + \mathrm{i} \frac{a d\Delta t }{\Delta x}\sin\bigg(\frac{2\pi j}{n_x}\bigg)
        \label{eqn: eigen_L_full}
    \end{equation}
and we compute these classically. 
$\mathcal{L}$ captures a single time step $\Delta t$. To evolve the solution to the chosen time $T$, we apply $\Lambda^{n_t}$.
Since $\Lambda^{n_t}$ is non-unitary, it must be block encoded within a unitary operator. We use the \ac{FABLE} subroutine to that end~\cite{Camps2022FABLE:Block-Encodings}.
It implements an approximate block encoding of any given matrix. This is achieved using $q + 1$ ancillary qubits by the circuit depicted in \cref{fig: fable example}, where $q = d\log_2 n_x$ is the number of qubits used to encode the state. The block-encoding produces the state $\ket{\psi_{n_t}}$ with success probability $||{\Lambda^{n_t}}||^2$. More specifically, the state of the circuit after stage (3) is 
\begin{equation}
    \ket{s_3} = \ket{0}^{\otimes(q +2)}\otimes \tilde{\Lambda} ^{n_t}\ket{\psi_{0}} + \sqrt{1-||\tilde{\Lambda} ^{n_t}||^2}\ket{\sigma^{\perp}} = \ket{0}^{\otimes(q +2)}\otimes \ket{\psi_{n_t}} + \sqrt{1-||\tilde{\Lambda} ^{n_t}||^2}\ket{0}\ket{\sigma^{\perp}}
    \label{eq:s3}
\end{equation}
where $\ket{\sigma^{\perp}}$ is a normalised state orthogonal to the ancilla register in the state $\ket{0}^{\otimes q}$. \ac{FABLE} is particularly useful because it implements an approximate block encoding. Since we do not require accuracy beyond the discretisation threshold set by ~\cref{eqn: def_epsc} 
, we choose to implement \ac{FABLE} such that
\begin{equation}
  ||\tilde{\Lambda}^{n_t}-\Lambda^{n_t}||_2 \leq \epsilon_c.
\end{equation}

The number of operations to implement the \ac{FABLE} subroutine is upper-bounded as $O(n_x^{2d})=O(4^q)$~\cite{Camps2022FABLE:Block-Encodings}. While the approximation to $\epsilon$ reduces the number of operations, the extent of this reduction remains an open question. We fix $\epsilon$ for our analysis; investigating its effect would make an interesting future research direction.  The contribution of the \ac{FABLE} subroutine to the overall circuit and analysis will be depicted in~\cref{fig:gate sets d=1and2} using \textcolor{FABLE}{orange triangles}.

\begin{figure}
\begin{quantikz}
    \lstick{flag: $\ket{0}$} &    \gategroup[3,steps=4,style={dashed,
 rounded corners,fill={FABLE!70}, inner xsep=2pt},
 background]{{\sc \acl{FABLE}}} &  \gate[2]{\mathrm{O_{\Lambda^{n_t}}} }  &    &      & \meter{} &  \rstick{$0$}\cw \\
    \lstick{zeros: $\ket{0}^{\otimes q}$} &  \gate{H^{\otimes n}}  &  \hphantom{\mathrm{O_{\Lambda^{n_t}}}}   &  \swap{1} &  \gate{H^{\otimes n}} &  \meter{} & \rstick{$0$} \cw \\
    \lstick{state: $\ket{\psi_0}$} &    &   &  \targX{} &   & \rstick{$\dfrac{{\Tilde{\Lambda}^{n_t}}\ket{\psi_0}}{||{\Lambda^{n_t}}\ket{\psi_0}||_2}$}  
\end{quantikz}
         \caption{The \ac{FABLE} circuit~\cite{Camps2022FABLE:Block-Encodings} from step (3) in~\cref{fig: AA_diag} for our problem. 
         It takes a state $\ket{\psi_0}$ and applies the matrix $\Lambda^{n_t}$ via block encoding, where $n_t$ is the number of time discretisation steps, where the state is encoded in $q$ qubits. The circuit width is $w = 2q+1$}
    \label{fig: fable example}
\end{figure}
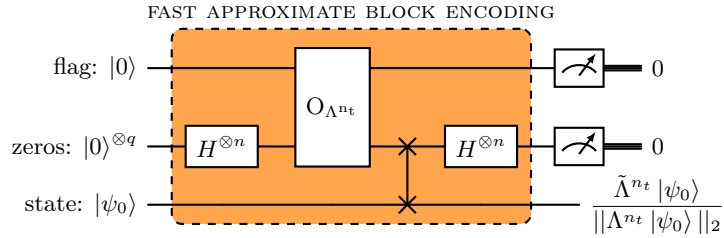

(4) 
To increase the probability of success of measuring the state we have prepared in the previous step, we must amplify the amplitude. 
Amplitude amplification \cite{Brassard2002QuantumEstimation} is a subroutine which can increase the probability of measuring the state or states of interest by amplifying the amplitude of those states. 
This is a simple approach that requires repeating every step of the circuit that was used to build the state of interest. An alternative subroutine is \ac{OAA}~\cite{BerryExponential}.
Here, we use an extra ancilla to perform \ac{OAA}; this allows us to rotate about an unknown state without repeating all the steps taken to build that state.
\ac{OAA} is a technique particularly suited to block encoding scenarios.  

\ac{OAA} applies the unitary operator $\mathcal{Q}_{OAA}$. 
We achieve this by utilising the \textit{AA flag} qubit, which is $\ket{1}$ only when the zero and flag registers of the \ac{FABLE} circuit are $\ket{0}$. Therefore, we do not require any knowledge of the state $\ket{\psi_{n_t}}$ to amplify it. We depict the circuit for \ac{OAA} in \cref{fig: OAA_diag}. The \ac{OAA} circuit block is repeated $j$ times where $ j = O(||\mathcal{L}^{n_t}|p_0\rangle||_2^{-1}) \leq (4\sqrt{n_t})^{d/2} \equiv N_{\text{OAA}}$ times~\cite[Lemma 12]{Devereux2025QuantumEquation}. 
The output state from this stage of the circuit is 
\begin{align}
    \ket{s_4} &= \mathcal{Q}_{OAA}^{N_{\text{OAA}}}\ket{s_3} \nonumber\\
    &= \sin(\dfrac{(2N_{\text{OAA}}+1)}{2}\theta)\ket{1}\ket{0}^{\otimes q+1}\ket{\psi_{n_t}} + \cos(\dfrac{(2N_{\text{OAA}}+1)}{2}\theta)\ket{0}\ket{\sigma ^{\perp}}\nonumber\\
    &= \alpha \ket{1}\ket{0}^{\otimes q+1}\ket{\psi_{n_t}} + \beta \ket{0}\ket{\sigma ^{\perp}},
    \label{eq:s4}
\end{align}
where, $\alpha = \sin(\dfrac{(2N_{\text{OAA}}+1)}{2}\theta)$, $\beta = \cos(\dfrac{(2N_{\text{OAA}}+1)}{2}\theta)$ and $\sin\theta/2 = \dfrac{1}{N_{\text{OAA}}}$.
The contribution of the \ac{OAA} subroutine to the overall circuit and analysis will be depicted in~\cref{fig:gate sets d=1and2} using \textcolor{OAA}{blue diamonds}.

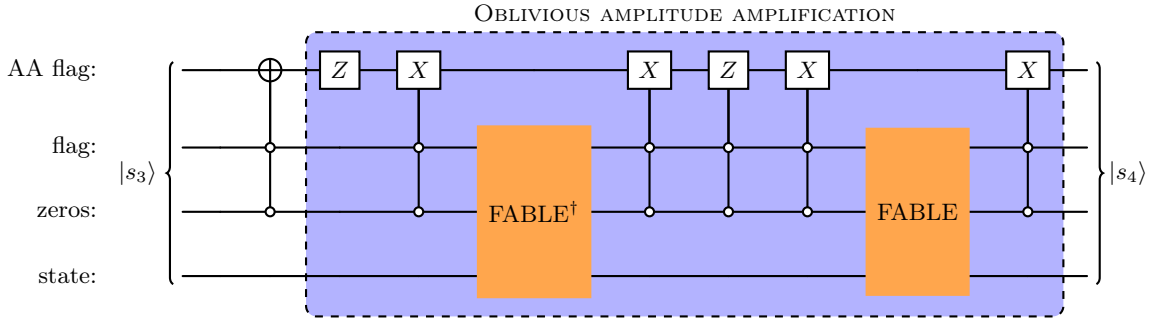
\begin{figure}
\begin{quantikz}
    \lstick{AA flag: } \setwiretype{n} & &\lstick[wires=4]{$\ket{s_3}$} & \setwiretype{q} & \targ{} & \gate[1]{\text{$Z$}}\gategroup[4,steps=8,style={dashed,
 rounded corners,fill={OAA!30}, inner xsep=2pt},
 background]{{\sc Oblivious amplitude amplification}} & \gate[1]{\text{$X$}} &   & \gate[1]{\text{$X$}} & \gate[1]{\text{$Z$}} & \gate[1]{\text{$X$}} &   & \gate[1]{\text{$X$}} &  \rstick[wires = 4]{$\ket{s_4}$} \\
    \lstick{flag: } \setwiretype{n} & & & \setwiretype{q} & \ctrl[open]{-1} &   & \ctrl[open]{-1} & \gate[3, style = {FABLE!70}]{\text{FABLE}^\dag} & \ctrl[open]{-1} & \ctrl[open]{-1} & \ctrl[open]{-1} & \gate[3, style = {FABLE!70}]{\text{FABLE}} & \ctrl[open]{-1} &   \\
    \lstick{zeros:}\setwiretype{n} & & & \setwiretype{q} &\ctrl[open]{-2}&   & \ctrl[open]{-2} & & \ctrl[open]{-2} & \ctrl[open]{-2} & \ctrl[open]{-2} & & \ctrl[open]{-2} &   \\
    \lstick{state:} \setwiretype{n} & & & \setwiretype{q} & &  &   & &   &   &   & &   &   
\end{quantikz}
    \caption{The quantum circuit implementing $\mathcal{Q}_{OAA}$ from step (4) in~\cref{fig: AA_diag}. It takes the state $\ket{s_3}$ from~\cref{eq:s3} as input and outputs the state $\ket{s_4}$ from~\cref{eq:s4}.}
    \label{fig: OAA_diag}
\end{figure}

(5) The penultimate stage of this algorithm is the \ac{IQFT}, which is the inverse of stage (2). The output of this stage is the state 
\begin{equation}
    \ket{s_5} = \mathcal{F}^{\dag}\ket{s_4} =  \alpha \ket{1}\ket{0}^{\otimes q+1}\ket{p_{n_t}} + \beta \ket{0}\mathcal{F}^{\dag}\ket{\sigma ^{\perp}}
    \label{eq:s5}
\end{equation}
where $|\tilde{p}_{n_t}\rangle = \frac{1}{||\tilde{p}(\mathbf{x}, n_t \Delta t)||_{2}} \sum_{\mathbf{x}}\tilde{p}(\mathbf{x},t = n_t \Delta t)|\mathbf{x},n_t \Delta t\rangle $.

(6) The final stage in our circuit is measurement. The measurement protocol in Ref.~\cite{Devereux2025QuantumEquation} follows Ref. \cite{VanApeldoorn2021QuantumEstimation}. However, this approach requires re-discretising the problem with a greater number of qubits, $\propto \log 1/\epsilon_q$, to reduce circuit depth. This qubit requirement exceeds the capabilities of current quantum devices or simulators. Consequently, we must rely on a more direct measurement protocol, using Hoeffding’s inequality to estimate the number of shots required.

\begin{lemma}[Hoeffding's inequality]
\label{lem: Hoeff}
    Suppose we are given a quantum algorithm that produces the state 
    \begin{equation}
        |\Tilde{\mathbf{p}}\rangle = \frac{1}{\sqrt{\sum_{(\mathbf{x},t=T)\in G}\Tilde{p}(\mathbf{x})^2}}\sum_{(\mathbf{x})\in G} \Tilde{p}(\mathbf{x})|\mathbf{x}\rangle,
    \end{equation}
    where $\mathbf{x} \in [0, n_x-1]^d$ for a given $t=T$. 
    To extract the probability distribution $\Tilde{p}(\mathbf{x})$ such that $||\Bar{\Tilde{p}}_N(\mathbf{x}) - \Tilde{p}(\mathbf{x})||_{\infty}\leq \epsilon_q$ with 
    $    \mathrm{Prob}(||\Bar{\Tilde{p}}_N(\mathbf{x}) - \Tilde{p}(\mathbf{x})||_{\infty}> \epsilon_q) \leq \delta
    $ 
   requires $N \geq \dfrac{-\log(\delta/2)}{2\epsilon_q^2}$ measurements  for each of the $n_x^d$ positions, 
    where $\Bar{\Tilde{p}}_N(\mathbf{x})$ is the sampled mean over $N$ samples and $1-\delta$ is the probability of success. 
\end{lemma}
The proof of this Lemma is provided 
in~\cref{app: Hoeff}. From here we will represent $\Bar{\Tilde{p}}_N(\mathbf{x}) $ as $\Tilde{\Tilde{p}}(\mathbf{x})$ as defined by~\cref{eqn: approx_definition,eqn: def_epsq}.

\label{sec: build}

\section{Results}
\label{sec: results}

\begin{table}
\centering
    \begin{tabular}{ c c c c c c c }
        \hline 
         $T$ [days] & $n_t$ & $a$ & $D$ & $\zeta$ [$\$^{-4}$] & $\epsilon$\\ 
         \hline \\
         10 & 100 & 0.2366 & 0.2455 & 1 & 0.03 \\  \\
         \hline
    \end{tabular}
    \caption{Fixed parameter values for DDE for a representative problem in finance \cite{Andersen2001TheVolatility, Wang2023ModelingProcess}. $T$ is the fixed time, which is discretised by $n_t$ steps. $a$ is the drift coefficient and represents the mean reversion parameter. $D$ is the diffusion coefficient and represents the volatility of the stock. $\zeta$ is a smoothness bound Ref.~\cite{Devereux2025QuantumEquation} and $\epsilon$ is defined in~\cref{eqn: approx_definition}.} 
    \label{tab: variables} 
\end{table}

For the construction and analysis of the quantum circuits, Quantinuum's tket development kit and compiler was used~\cite{Sivarajah2020Tket:Devices}. Circuit simulations were executed via the Qulacs quantum simulator backend~\cite{Suzuki2021Qulacs:Purpose}.
Tket facilitates circuit design, visualisation, and analysis, including the computation of circuit depth across various gate sets. This is enabled by tket’s compilation capabilities, which allow for consistent and optimised circuit transformations. Circuit composition can be influenced by multiple factors, so in our study, all reported metrics are derived from circuits compiled under uniform conditions. As such, the numerical results are considered accurate to within an order of magnitude. Further discussion of best circuit decomposition can be found in Ref.~\cite{Sivarajah2020Tket:Devices}. 
\cref{tab:scenarios} provides a description of each scenario used in our circuit analysis. We employed up to $22$ qubits and reached approximately $1\times 10^6$ operations. 
We deem this set of scenarios sufficient to demonstrate the trends and extract our main conclusions. As we discuss in~\cref{sec:discuss}, $1\times 10^6$ operations exceeds the capabilities of current quantum computers.

\begin{table}[]
\centering
\begin{tabular}{c c c c c c }
\hline
\textbf{Scenario} & Result & $n_x$ & $w$ & $L$ & $d$ \\
\hline
Simulation results comparison& \cref{fig: circuit_output_1d} & $32$ & $12$ & $20$ & $1$ \\
" & \cref{fig: circuit_output_1d_nx64} & $64$ & $14$ & $20$ & $1$ \\
" & \cref{fig: 2d_slices} & $16$ & $18$ & $40$ & $2$ \\

Error analysis & \cref{fig: n_x 1d} & $16-128$ & $-$ & $60$ & $1$\\

Gate set depth analysis: \ac{QFT} & \cref{Fig: gate sets qft} & $4-32$ & $2-5$ & $40$ & $1$\\

Gate set depth analysis: \ac{FABLE} & \cref{Fig: gate sets FABLE} & $4-32$ & $5-11$ & $40$ & $1$\\

Circuit depth analysis & \cref{Fig: gate sets_d=1} & $16-128$ & $10-16$ & $60$ & $1$\\

" & \cref{Fig: gate sets_d=2} & $8-32$ & $14-22$ & $60$ & $2$ \\

Resource analysis & \cref{tab:resource_estimation} & $16-32$ & $10-22$ & $60$ & $1-2$ \\

State preparation method comparison & \cref{app: SP} & $16-64$ & $4-12$ & $60$ & $1-3$ \\

\hline
\end{tabular}
\caption{Scenarios and parameters that have been tested within this paper. $n_x$, is the number of spatial discretisation steps, $w$ is the total circuit width, $L$ is the limit of the spatial dimensions and $d$ is the number of spatial dimensions. All other parameters are as listed in~\cref{tab: variables}.}
\label{tab:scenarios}
\end{table}

To validate our quantum circuit, we first analyse its output using a quantum simulator~\cite{Suzuki2021Qulacs:Purpose}. 
Our results are compared to both the analytical solution from \cref{eqn: analytic} and the classical solution method from \cref{thm: Classical_FFT} to confirm their accuracy. Initially, we begin our analysis by simulating one spatial dimension. The fixed parameters for this test are outlined in \cref{tab: variables}. These values were chosen due to their relevance to financial stock volatility~\cite{Andersen2001TheVolatility, Wang2023ModelingProcess}. In~\cref{fig: 1d_results} we demonstrate the circuit output in two scenarios; when $n_x = 32$ and $n_x = 64$. In both cases $L = 20$, all other parameters are as found in~\cref{tab: variables}. We also simulate the circuit for $d=2$ and $n_x = 16$ and share the results in \cref{app: d2}. We find in all cases that the quantum circuit matches the classical and analytic solutions to within $\epsilon$.

\begin{figure}
\floatsetup{valign=t, heightadjust=object}
\ffigbox{%
\begin{subfloatrow}
\ffigbox{\scalebox{0.85}{
\begin{tikzpicture}

\definecolor{crimson2143940}{RGB}{214,39,40}
\definecolor{darkgray176}{RGB}{176,176,176}
\definecolor{darkorange25512714}{RGB}{255,127,14}
\definecolor{forestgreen4416044}{RGB}{44,160,44}
\definecolor{lightgray204}{RGB}{204,204,204}
\definecolor{steelblue31119180}{RGB}{31,119,180}

\begin{axis}[
legend cell align={left},
legend style={fill opacity=0.8, draw opacity=1, text opacity=1, draw=lightgray204},
legend pos = north west,
tick align=outside,
tick pos=left,
x grid style={darkgray176},
xlabel={x},
xmajorgrids,
xmin=-22, xmax=22,
xtick style={color=black},
y grid style={darkgray176},
ylabel={p(x)},
ymajorgrids,
ymin=-0.0205474841201229, ymax=0.431497166522582,
ytick style={color=black}
]
\addplot [semithick, blue]
table {%
-20 3.27663350075914e-106
-18.7096774193548 3.03565130891006e-94
-17.4193548387097 5.32114342023137e-83
-16.1290322580645 1.76476591376352e-72
-14.8387096774194 1.10738417120691e-62
-13.5483870967742 1.31473607142464e-53
-12.258064516129 2.95330232928923e-45
-10.9677419354839 1.25518083564828e-37
-9.67741935483871 1.00933132956714e-30
-8.38709677419355 1.5356430215252e-24
-7.09677419354839 4.42054534150651e-19
-5.80645161290323 2.4076341283397e-14
-4.51612903225806 2.48104244764297e-10
-3.2258064516129 4.8373438672255e-07
-1.93548387096774 0.000178446810222584
-0.64516129032258 0.0124548838026755
0.64516129032258 0.164474900923117
1.93548387096774 0.410949682402459
3.2258064516129 0.194270441527807
4.51612903225806 0.01737616243112
5.80645161290322 0.000294056017088748
7.09677419354839 9.41532543572158e-07
8.38709677419355 5.70387058682848e-10
9.67741935483871 6.53782141576549e-14
10.9677419354839 1.41783438910406e-18
12.258064516129 5.81764237549835e-24
13.5483870967742 4.5164540672957e-30
14.8387096774194 6.6340270399171e-37
16.1290322580645 1.84368167584679e-44
17.4193548387097 9.69446179892705e-53
18.7096774193548 9.64474012912171e-62
20 1.81545889197598e-71
};
\addlegendentry{Initial condition}
\addplot [semithick, darkorange25512714]
table {%
-20 0
-18.7096774193548 0
-17.4193548387097 0
-16.1290322580645 0.00534560995293115
-14.8387096774194 0.00534560995293115
-13.5483870967742 0
-12.258064516129 0
-10.9677419354839 0
-9.67741935483871 0.00534560995293115
-8.38709677419355 0
-7.09677419354839 0.00925886803592263
-5.80645161290323 0.0130940167486245
-4.51612903225806 0.0410604091597212
-3.2258064516129 0.0781995319329955
-1.93548387096774 0.130502970380947
-0.64516129032258 0.164762764241296
0.64516129032258 0.15685517175727
1.93548387096774 0.109812933275029
3.2258064516129 0.0489933244900326
4.51612903225806 0.0185177360718453
5.80645161290322 0
7.09677419354839 0
8.38709677419355 0.00534560995293115
9.67741935483871 0
10.9677419354839 0.00755983409459183
12.258064516129 0
13.5483870967742 0
14.8387096774194 0
16.1290322580645 0
17.4193548387097 0
18.7096774193548 0
20 0
};
\addlegendentry{Quantum solution}
\addplot [semithick, forestgreen4416044]
table {%
-20 8.21180034170394e-10
-18.7096774193548 5.91670153145471e-09
-17.4193548387097 3.96417176675698e-08
-16.1290322580645 2.46068779953239e-07
-14.8387096774194 1.40921857238137e-06
-13.5483870967742 7.41052696238697e-06
-12.258064516129 3.55869504907565e-05
-10.9677419354839 0.000155076939422081
-9.67741935483871 0.000608686143068881
-8.38709677419355 0.00213311907659649
-7.09677419354839 0.00660459549557833
-5.80645161290323 0.0178379954911223
-4.51612903225806 0.0413703732732972
-3.2258064516129 0.0807855519732288
-1.93548387096774 0.129544833470647
-0.64516129032258 0.165196478018054
0.64516129032258 0.160801599776597
1.93548387096774 0.113777078863085
3.2258064516129 0.0560149032575911
4.51612903225806 0.0192287198003419
5.80645161290322 0.00481335544992947
7.09677419354839 0.00092231801373991
8.38709677419355 0.000140859394727347
9.67741935483871 1.76968765790844e-05
10.9677419354839 1.87459448899592e-06
12.258064516129 1.70701682733552e-07
13.5483870967742 1.35712942924123e-08
14.8387096774194 9.53972469281541e-10
16.1290322580645 6.00626909667706e-11
17.4193548387097 4.8542055819797e-12
18.7096774193548 1.30560441006932e-11
20 1.0633641716387e-10
};
\addlegendentry{Classical solution}
\addplot [semithick, crimson2143940]
table {%
-20 7.51827958765845e-16
-18.7096774193548 4.9732945982729e-14
-17.4193548387097 2.47610002491577e-12
-16.1290322580645 9.27877124735372e-11
-14.8387096774194 2.61704407232254e-09
-13.5483870967742 5.55558619022206e-08
-12.258064516129 8.87660772206031e-07
-10.9677419354839 1.06748664601366e-05
-9.67741935483871 9.66220302312856e-05
-8.38709677419355 0.000658245882889924
-7.09677419354839 0.0033751915805474
-5.80645161290323 0.0130258775569737
-4.51612903225806 0.0378367511651933
-3.2258064516129 0.0827216085909816
-1.93548387096774 0.136120140324867
-0.64516129032258 0.168587002955284
0.64516129032258 0.157153506939178
1.93548387096774 0.110261107050263
3.2258064516129 0.0582262613185032
4.51612903225806 0.023142685131429
5.80645161290322 0.00692319992070054
7.09677419354839 0.00155882830856319
8.38709677419355 0.000264172727448282
9.67741935483871 3.3695814680549e-05
10.9677419354839 3.2349101810696e-06
12.258064516129 2.33747393835364e-07
13.5483870967742 1.27124603953274e-08
14.8387096774194 5.20368181536729e-10
16.1290322580645 1.6032089865159e-11
17.4193548387097 3.71764424915846e-13
18.7096774193548 6.48849245066468e-15
20 8.52350192417156e-17
};
\addlegendentry{Analytical solution}
\end{axis}

\end{tikzpicture}
}}{\caption{
 \label{fig: circuit_output_1d}}}
\ffigbox{\scalebox{0.85}{
 \input{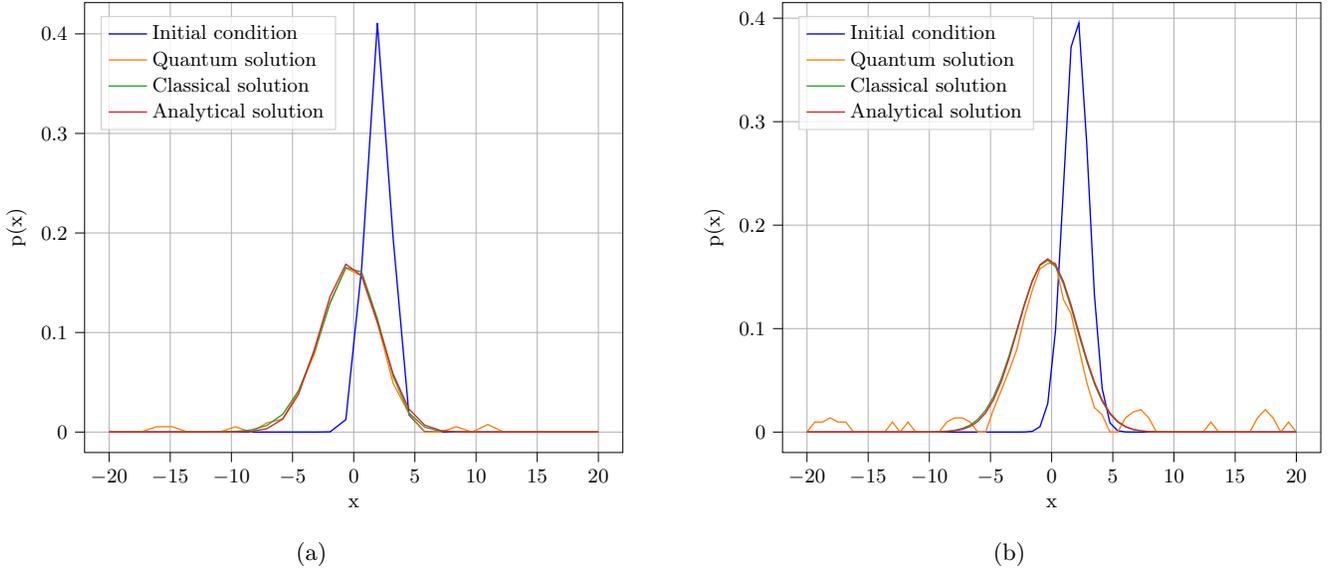}}}{\caption{\label{fig: circuit_output_1d_nx64}}}
\end{subfloatrow}}
{\caption{The $d=1$ simulation comparison of the classical and quantum diagonalisation by Fourier transform methods from \cref{thm: QFT} and \cref{thm: Classical_FFT} for the scenario described by~\cref{tab: variables} when $L=20$. The circuits were built using tket and simulated using Qulacs~\cite{Sivarajah2020Tket:Devices,Suzuki2021Qulacs:Purpose}. (a) The output of the circuit for $n_x = 32$, which uses 12 qubits, and (b) for $n_x = 64$, which uses 14 qubits.}\label{fig: 1d_results}}
\end{figure}

To further inspect the accuracy of these results, we plot how the error contributions scale with $n_x$ in \cref{fig: n_x 1d}. 
We determine $\epsilon_c$ by comparison between the analytical and classical solution, since $\epsilon_c$ is a purely classical error from discretisation. Then, we can calculate $\epsilon_q$ by comparison between the classical and quantum solutions, as this is the error in approximating the distribution using quantum measurement. Finally, $\epsilon = \epsilon_c + \epsilon_q$. As stated in \cref{tab: variables}, our target for $\epsilon$ is $ 0.03$. We can see that for $N=50,000$, we only meet the target $\epsilon$ for $n_x = 32$ and $n_x = 64$. These results show that increasing $n_x$ reduces the discretisation error $\epsilon_c$, as is expected. However, increasing $n_x$ increases the number of qubits and the number of gates. Consequently, $\epsilon_q$ increases for a fixed number of shots $N$. Reducing the total error will always be a balance between these two competing factors. This was a key result in Ref.~\cite{Devereux2025QuantumEquation} that~\cref{fig: n_x 1d} has demonstrated in practice. 

\begin{figure}
\floatsetup{valign=c, heightadjust=object}
\begin{floatrow}
\ffigbox{\scalebox{0.8}{\begin{tikzpicture}
\begin{axis}[
    xlabel={$n_x$},
    ylabel={$\epsilon$ [$10^{-2}$]},
    xmin=8, xmax=256,
    ymin=0, ymax=7,
    xtick={8, 16, 32, 64, 128, 256},
    xticklabels={8, 16, 32, 64, 128, 256},
    ytick={0, 1, 2, 3, 4, 5, 6, 7},
    legend pos= north east,
    ymajorgrids=true,
    xmajorgrids=true,
    xmode=log,
    log basis x={2}
]

\addplot[
    color=blue,
    mark=x,
    ]
    coordinates {
    (16, 6.14)(32, 2.9)(64, 2.9)(128, 3.957)
    };
\addplot[
    color=red,
    mark=x,
    ]
    coordinates {
    (16, 6)(32, 2.7)(64, 1.7)(128, 0.26)
    };
\addplot[
    color=orange,
    mark=x,
    ]
    coordinates {
    (16, 0.14)(32, 0.2)(64, 1.2)(128, 3.7)
    };
\addplot[
    color=black,
    no markers,
    dashed
    ]
    coordinates {
    (8, 3)(256, 3)
    };
    \legend{$\epsilon$, $\epsilon_c$, $\epsilon_q$}
\end{axis}
\begin{axis}[
        xmin=1, xmax=6,
        xtick={1, 2, 3, 4, 5, 6},
        xticklabels={8, 10, 12, 14, 16, 18},
        ymin=0, ymax=70,
        hide y axis,
        axis x line*=top,
        xlabel={Qubits}
    ]
  \end{axis}
\end{tikzpicture}}}{\caption{An analysis of the error contributions that shows the variation of $\epsilon$ with respect to the number of spatial discretisation points, $n_x$, for $N= 50,000$.}\label{fig: n_x 1d}}
\capbtabbox{\scalebox{0.85}{
    \begin{tabular}{ l c }
        \hline
        \vspace{2pt}
         Label & Gates \\ 
         \hline
         \vspace{2pt}
         IBM Heron~\cite{ProcessorDocumentation} & $\sqrt{X}$, $R_z(\theta)$, $X$, $CZ$ \\ 
         Fujitsu \acs{STAR}~\cite{Toshio2025PracticalComputer}& $CNOT$, $H$, $S$ $R_z(\theta)$ \\
         Quantinuum Tket~\cite{Sivarajah2020Tket:Devices} & $Tk1(\alpha, \beta, \gamma) = R_z(\alpha)R_y(\beta)R_z(\gamma)$, $CCX$ \\
         QFT~\cite{Nielson2010QuantumInformation} & $CX$, $R_z(\theta)$,  $SWAP$\\
         IonQ~\cite{NativeDocumentation} & $GPi(\theta)= R_z(-\theta)XR_z(\theta)$, $R_z(\pi/2)\, \textrm{or}\, R_y(\pi/2)$, $Z$\\
         \hline
    \end{tabular}}}{\caption{List of gate sets for a variety of hardware.} \label{tab: gate sets}}
    \end{floatrow}
\end{figure}

Throughout this analysis, we use $N = 50,000$. This was chosen based on observing a plateau in $\epsilon_q$ when varying $N$. This can be seen in~\cref{app: Hoeff}. From \cref{fig: n_x 1d}, we know that for $n_x = 32$  $\epsilon_c = 0.027$ and we desire an $\epsilon = 0.03$. Therefore, we require $\epsilon_q \leq 0.003$. In Lemma~\ref{lem: Hoeff}, we introduced an upper bound for the number of shots required to find $\Tilde{\Tilde{p}}(\mathbf{x})$ to within $\epsilon_q$ of $\Tilde{p}(\mathbf{x})$. This bound ensures a success probability of $1-\delta$. Using Lemma~\ref{lem: Hoeff}, the plateau observed is equivalent to a success probability of $1-\delta = 0.81$.

\section{Depth scaling}
\label{Sec: res_est}

Having defined the scenarios and verified the accuracy of the corresponding quantum circuits, we now analyse how the circuit depth scales and compare it against the theoretical complexity scaling. We will test how each circuit stage varies with dimension, $d$, and the number of spatial discretisation steps, $n_x$.   
To ensure a meaningful analysis, we first establish a consistent, universal gate set.
Various quantum computing gate sets exist, broadly categorised as theoretically universal or native to current devices. We begin by performing a gate set analysis to select the most suitable one for our subsequent analysis.

Most algorithmic analysis considers a very broad set of gates, e.g., all known 1 and 2-qubit gates, which we will refer to as an unconstrained gate set in the remainder of the text~\cite{Linden2022QuantumEquation, Plesch2011Quantum-stateDecompositions, Arrazola2019QuantumEquations}. Alternatively, a limited and theoretical universal gate set can be chosen, such as Toffoli ($CCX$) and Hadamard ($H$)~\cite{Nielson2010QuantumInformation}. Neither of these is particularly useful when considering implementing algorithms on real hardware. Therefore, we have completed a representative analysis of how the \ac{QFT} and \ac{FABLE} circuits scale with qubit number for a variety of hardware specific native gate sets. The first two are superconducting gate sets. The first from Fujitsu is the \ac{STAR} gate set~\cite{Toshio2025PracticalComputer}, the second is the IBM Heron gate set \cite{ProcessorDocumentation}. We also included the IonQ trapped ion gate set~\cite{NativeDocumentation}. Finally, we also included the Quantinuum TK1 gate set, which is not specific to a hardware type and uses just two gates, including the three-qubit Toffoli gate~\cite{Sivarajah2020Tket:Devices}. The gate sets are summarised in \cref{tab: gate sets}.

\cref{fig:gate sets} illustrates the circuit depth results for two subroutines, categorised by the gate sets listed in \cref{tab: gate sets}. We begin with the \ac{QFT} subroutine, given its well-documented and easily understood structure. Practically, the theoretical and unconstrained circuit depths for \ac{QFT} are very similar. This similarity is straightforward to demonstrate: a \ac{QFT} circuit for $k$ qubits can be constructed with $k$ Hadamard gates, $k/2$ swap gates and $k(k-1)/2$ controlled unitary rotations, yielding a gate count of $O(k^2)$~\cite{Nielson2010QuantumInformation}. Our unconstrained implementation reflects this. Although the gate count is $O(k^2)$, the theoretical and unconstrained circuit depths effectively match because many of these gates can be executed in parallel, reducing the depth. The unconstrained gate set yields the lowest circuit depth, essentially by performing multiple gates simultaneously. 

However, this ideal scenario does not accurately represent how quantum computers operate in practice. The \ac{STAR} gate set represents the most efficient practical implementation. Though there are challenges with this gate set when it comes to efficient analogue rotations~\cite{Toshio2025PracticalComputer}. In contrast, the IonQ gate set analysis resulted in a significantly higher gate count. To maintain clarity, we present a truncated set of these results.
In \cref{Fig: gate sets FABLE}, a similar analysis was performed for the \ac{FABLE} subroutine, which, as shown in \cref{fig:gate sets d=1and2}, contributes most to the overall algorithm's circuit depth. 
\cref{Fig: gate sets FABLE} reveals a steep rise in circuit depth for the \ac{FABLE} subroutine as the number of qubits grows, particularly under the IBM gate set. This is consistent with the theoretical expectation of \ac{FABLE} scaling as $O(n_x^{2d}$).
Here again, the tket and \ac{STAR} gate sets represent the most efficient implementations. Interestingly, the IBM gate set proves least efficient in this instance. This is likely because the \ac{FABLE} subroutine primarily consists of rotation gates, a type of operation for which the IonQ gate set is better suited than the IBM gate set.

The choice of gate set for a practical implementation is complicated. A perfectly implementable gate set, if it existed, would be universally adopted. Implementing gates efficiently is an active area of research across hardware types. For example, each IBM machine has a different native gate set~\cite{ProcessorDocumentation}. The Fujitsu gate set does not yet have a way to suppress errors escalating for high rotation circuits~\cite{Toshio2025PracticalComputer}.
Trapped-ion gate counts are incredibly high in both the \ac{QFT} and \ac{FABLE} scenarios in \cref{fig:gate sets}. The alternative would be to choose a theoretical universal gate set. This is simpler in theory, but of limited practical value. 
The tket gate set serves as a useful intermediate solution. While not a standard universal gate set, it is universal. 
All of the backends in tket are decomposed via this gate set. 
Therefore, it provides an easily testable measure for every circuit.

Our analysis uses the tket and \ac{STAR} gate sets. This is driven by their ability to achieve the lowest circuit depths in~\cref{fig:gate sets d=1and2} and allows us to present both theoretical and practical results.

\begin{figure}
\floatsetup{valign=t, heightadjust=object}
\ffigbox{%
\begin{subfloatrow}
\ffigbox{\scalebox{0.8}{\begin{tikzpicture}
\begin{axis}[
    xlabel={Qubits},
    ylabel={Depth},
    xmin=1, xmax=6,
    ymin=0, ymax=70,
    xtick={0, 1, 2, 3, 4, 5, 6},
    ytick={0, 10, 20, 30, 40, 50, 60, 70},
    legend pos= north east,
    ymajorgrids=true,
    xmajorgrids=true,
    title style={yshift=5ex},
    axis x line*=top,
]

\addplot[
    color=Cerulean,
    mark=x,
    ]
    coordinates {
    (2,4)(3,6)(4,8)(5,10)
    };
\addplot[
    color=red,
    mark=x,
    ]
    coordinates {
    (2,4)(3,8)(4,12)(5,18)
    };
\addplot[
    color=ForestGreen,
    mark=x,
    ]
    coordinates {
    (2,5)(3,12)(4,21)(5,28)
    };
\addplot[
    color=magenta,
    mark=x,
    ]
    coordinates {
    (2,8)(3,16)(4,24)(5,32)
    };
\addplot[
    color=YellowOrange,
    mark=x,
    ]
    coordinates {
    (2,12)(3,20)(4,28)(5,36)
    };
\addplot[
    color=blue,
    mark=x,
    ]
    coordinates {
    (2,33)(3,67)(4,101)(5,135)
    };
    
    \legend{Unconst., Theoretical, Tket, STAR, IBM, IonQ}
    
\end{axis}
\begin{axis}[
        xmin=1, xmax=6,
        xtick={1, 2, 3, 4, 5, 6},
        xticklabels={2, 4, 8, 16, 32, 64},
        ymin=0, ymax=70,
        hide y axis,
        axis x line*=bottom,
        xlabel={$n_x$}
    ]
  \end{axis}
\end{tikzpicture}}}{\caption{}\label{Fig: gate sets qft}}
\ffigbox{\scalebox{0.8}{
\begin{tikzpicture}
\begin{axis}[
    xlabel={Qubits},
    ylabel={Depth},
    xmin=1, xmax=6,
    ymin=0, ymax=1200,
    xtick={ 1, 2, 3, 4, 5, 6},
    xticklabels={3, 5, 7, 9, 11, 13},
    ytick={0, 200, 400, 600, 800, 1000, 1200},
    legend pos= north west,
    ymajorgrids=true,
    xmajorgrids=true,
    title style={yshift=5ex},
    axis x line*=top,
]

\addplot[
    color=red,
    mark=x,
    ]
    coordinates {
    (2,16)(3,64)(4,256)(5,1024)
    };
\addplot[
    color=ForestGreen,
    mark=x,
    ]
    coordinates {
    (2,16)(3,50)(4,98)(5,194)
    };
\addplot[
    color=magenta,
    mark=x,
    ]
    coordinates {
    (2,32)(3,84)(4,164)(5,324)
    };
\addplot[
    color=YellowOrange,
    mark=x,
    ]
    coordinates {
    (2,84)(3,273)(4,537)(5,1065)
    };
\addplot[
    color=blue,
    mark=x,
    ]
    coordinates {
    (2,57)(3,180)(4,340)(5,660)
    };
    
    \legend{Theoretical, Tket, STAR, IBM, IonQ}
    
\end{axis}
\begin{axis}[
        xmin=1, xmax=6,
        xtick={1, 2, 3, 4, 5, 6},
        xticklabels={2, 4, 8, 16, 32, 64},
        ymin=0, ymax=70,
        hide y axis,
        axis x line*=bottom,
        xlabel={$n_x$}
    ]
\end{axis}
\end{tikzpicture}}}{\caption{} \label{Fig: gate sets FABLE}}
    \end{subfloatrow}}
    {\caption{An analysis of the circuit depth scaling in five gate sets of (a) the \ac{QFT} subroutine and (b) the \ac{FABLE} subroutine, dependent on the number of qubits. This analysis was completed for the gate sets described in \cref{tab: gate sets} and the theoretical scaling of the subroutines. This set of scenarios, whilst small, is enough to demonstrate the trend.}
    \label{fig:gate sets}}
\end{figure}
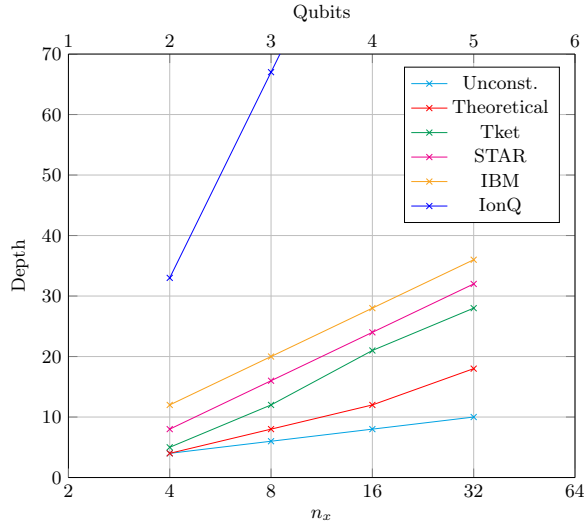
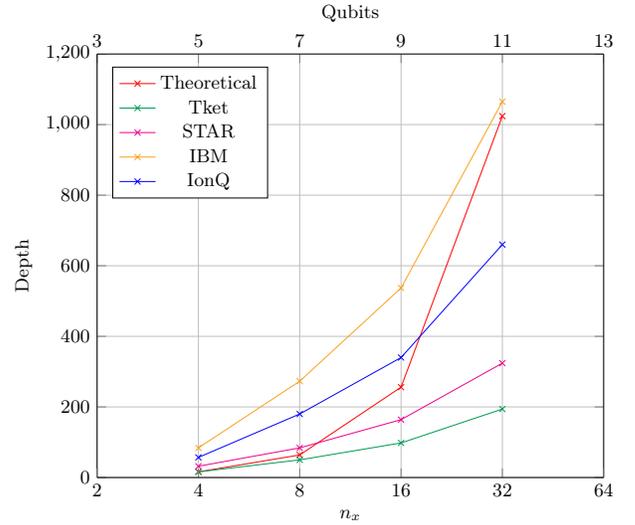

Within a chosen gate set, we can consider how the circuit depth scales with a variety of variables. We begin by considering how each of the subroutines scales with $d$ and $n_x$. 
To do this, we employ the chosen basis gate sets, considering a variety of $n_x$ values for $d=1$, as can be seen in \cref{Fig: gate sets_d=1}. 
Based on \cref{thm: QFT}, we expect the total scaling to go as $O(d^2 n_x^{d/2})$. \cref{Fig: gate sets_d=1} shows that our total scaling matches well with the theoretical scaling. 
The highest contributing subroutine to the overall gate count is the \ac{FABLE} subroutine (the \ac{OAA} subroutine depends on the \ac{FABLE} subroutine). 
Considering how the gate set affects the scaling, we find that the tket gate set does have a lower overall gate count than the \ac{STAR} gate set. This is due to the use of a 3-qubit gate in the former. The state preparation also scales significantly better in the tket gate set than the \ac{STAR} gate set. The \ac{QFT} scales very similarly between the gate sets. The \ac{FABLE} subroutine scaling is similar across the two gate sets; the power of the 3-qubit gate in the tket gate set is offset by the analogue rotation gate in the \ac{STAR} gate set. However, the \ac{STAR} gate set has a higher overhead for the \ac{FABLE} subroutine. All of this contributes to the overall scaling being better in the tket gate set. 

In~\cref{Fig: gate sets_d=2} we also consider the effects of dimension on this scaling by studying the same plot for $d=2$. When $d=1$, the theoretical scaling closely matches the total scaling. However, for $d=2$ the total circuit depth is higher than the theoretical scaling. Furthermore, increasing $n_x$ gives a steeper scaling than theorised. This discrepancy between theorised and actual is due to \ac{FABLE} scaling as $4^q = n_x ^{2d} = n_x ^4$. Whereas, in Ref.~\cite{Devereux2025QuantumEquation, Linden2022QuantumEquation}, this step was theorised to have scaling as $\Tilde{O}(n_x)$.
The \ac{QFT} circuit depth is unchanged with $d$, as increased dimension requires simultaneous application of a \ac{QFT} per dimension. The largest contributor to increased circuit depth is the \ac{FABLE} circuit depth. The total circuit depth scaling with $d=2$ is an order of magnitude larger than for $d=1$. 

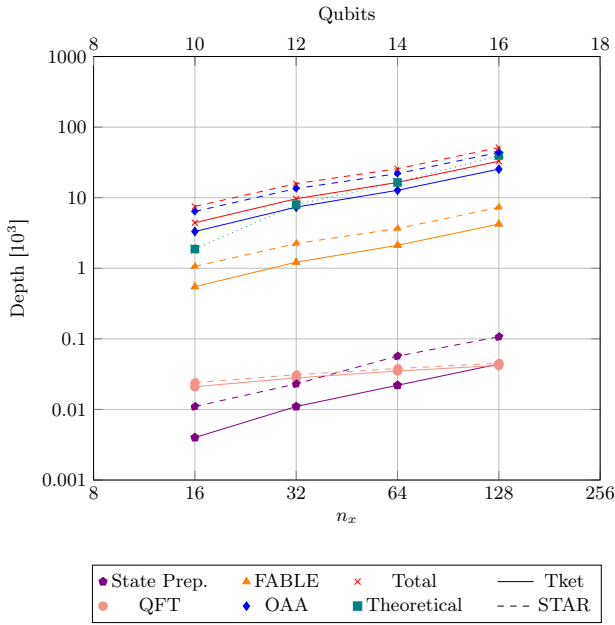
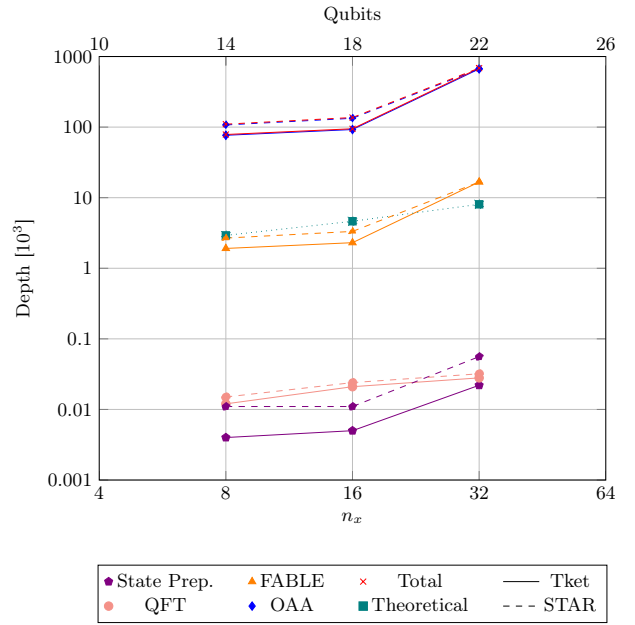
\begin{figure}
\floatsetup{valign=t, heightadjust=object}
\ffigbox{%
\begin{subfloatrow}
\ffigbox{\scalebox{0.8}{\begin{tikzpicture}

\begin{axis}[
    title style={yshift=5ex},
    xlabel={$n_x$},
    ylabel={Depth [$10^3$]},
    xmin=8, xmax=256,
    ymin=0.001, ymax=1000,
    xtick={8, 16, 32, 64, 128, 256},
    xticklabels={8, 16, 32, 64, 128, 256},
    ytick={0.001, 0.01, 0.1, 1, 10, 100, 1000},
    yticklabels={0.001, 0.01, 0.1, 1, 10, 100, 1000},
    transpose legend,
    legend style={at={(0.5,-0.2)},anchor=north},
    legend style={/tikz/every even column/.append style={column sep=0.5cm}},
    legend columns=2,
    ymajorgrids=true,
    xmajorgrids=true,
    xmode=log,
    log basis x={2},
    ymode = log,
    log basis y={10}
]
\addlegendimage{only marks, mark=pentagon*, color=SP}
\addlegendentry{State Prep.}
\addlegendimage{only marks, mark=*, color= QFT}
\addlegendentry{QFT}
\addlegendimage{only marks, mark=triangle*, color=FABLE}
\addlegendentry{FABLE}
\addlegendimage{only marks, mark=diamond*, color = OAA}
\addlegendentry{OAA}
\addlegendimage{only marks, mark=x, color=red}
\addlegendentry{Total}
\addlegendimage{only marks, mark=square*, color= teal}
\addlegendentry{Theoretical}
\addlegendimage{no markers, color= black}
\addlegendentry{Tket}
\addlegendimage{no markers, dashed, color= black}
\addlegendentry{STAR}
\addplot[
    color=SP,
    mark=pentagon*,
    ]
    coordinates {
    (16,0.004)(32,0.011)(64,0.022)(128,0.044)
    };
\addplot[
    color=QFT,
    mark=*,
    ]
    coordinates {
    (16,0.021)(32,0.028)(64,0.035)(128,0.042)
    };

\addplot[
    color=FABLE,
    mark=triangle*,
    ]
    coordinates {
    (16,0.548)(32,1.220)(64,2.116)(128,4.228)
    };

\addplot[
    color=OAA,
    mark=diamond*,
    ]
    coordinates {
    (16,3.319)(32,7.351)(64,12.727)(128,25.399)
    };

\addplot[
    color=red,
    mark=x,
    ]
    coordinates {
    (16,4.425)(32,9.662)(64,16.471)(128,32.827)
    };
\addplot[
    color=teal,
    mark=square*,
    dotted,
    ]
    coordinates {
    (16,1.872)(32,7.979)(64,16.463)(128,39.646)
    };

\addplot[
    color=SP,
    mark=pentagon*,
    dashed,
    ]
    coordinates {
    (16, 0.011)(32,0.023)(64,0.057)(128,0.107)
    };
\addplot[
    color=QFT,
    mark=*,
    dashed,
    ]
    coordinates {
    (16,0.024)(32,0.031)(64,0.038)(128,0.045)
    };
\addplot[
    color=FABLE,
    mark=triangle*,
    dashed,
    ]
    coordinates {
    (16,1.060)(32,2.244)(64,3.652)(128,7.3)
    };
\addplot[
    color=OAA,
    mark=diamond*,
    dashed,
    ]
    coordinates {
    (16,6.391)(32,13.495)(64,21.943)(128,43.831)
    };
\addplot[
    color=red,
    mark=x,
    dashed,
    ]
    coordinates {
    (16,7.51)(32,15.824)(64,25.728)(128,51.328)
    };
\end{axis}
\begin{axis}[
        xmin=1, xmax=6,
        xtick={1, 2, 3, 4, 5, 6},
        xticklabels={8, 10, 12, 14, 16, 18},
        ymin=0, ymax=70,
        hide y axis,
        axis x line*=top,
        xlabel={Qubits}
    ]
\end{axis}

\end{tikzpicture}}}{\caption{}\label{Fig: gate sets_d=1}}
\ffigbox{\scalebox{0.8}{
\begin{tikzpicture}

\begin{axis}[
    title style={yshift=5ex},
    xlabel={$n_x$},
    ylabel={Depth [$10^3$]},
    xmin=4, xmax=64,
    ymin=0.001, ymax=1000,
    xtick={4, 8, 16, 32, 64},
    xticklabels={4, 8, 16, 32, 64},
    ytick={0.001, 0.01, 0.1, 1, 10, 100, 1000},
    yticklabels={0.001, 0.01, 0.1, 1, 10, 100, 1000},
    transpose legend,
    legend style={at={(0.5,-0.2)},anchor=north},
    legend style={/tikz/every even column/.append style={column sep=0.5cm}},
    legend columns=2,
    ymajorgrids=true,
    xmajorgrids=true,
    xmode=log,
    log basis x={2},
    ymode = log,
    log basis y={10}
]
\addlegendimage{only marks, mark=pentagon*, color=SP}
\addlegendentry{State Prep.}
\addlegendimage{only marks, mark=*, color= QFT}
\addlegendentry{QFT}
\addlegendimage{only marks, mark=triangle*, color=FABLE}
\addlegendentry{FABLE}
\addlegendimage{only marks, mark=diamond*, color = OAA}
\addlegendentry{OAA}
\addlegendimage{only marks, mark=x, color=red}
\addlegendentry{Total}
\addlegendimage{only marks, mark=square*, color= teal}
\addlegendentry{Theoretical}
\addlegendimage{no markers, color= black}
\addlegendentry{Tket}
\addlegendimage{no markers, dashed, color= black}
\addlegendentry{STAR}
\addplot[
    color=SP,
    mark=pentagon*,
    ]
    coordinates {
    (8,0.004)(16,0.005)(32,0.022)
    };
\addplot[
    color=QFT,
    mark=*,
    ]
    coordinates {
    (8,0.012)(16,0.021)(32,0.028)
    };

\addplot[
    color=FABLE,
    mark=triangle*,
    ]
    coordinates {
    (8,1.91)(16,2.304)(32,16.686)
    };

\addplot[
    color=OAA,
    mark=diamond*,
    ]
    coordinates {
    (8,76.6)(16,92.361)(32,667.641)
    };

\addplot[
    color=red,
    mark=x,
    ]
    coordinates {
    (8,78.539)(16,94.712)(32,684.405)
    };
\addplot[
    color=teal,
    mark=square*,
    dotted,
    ]
    coordinates {
    (8,2.924)(16,4.636)(32,8.061)
    };

\addplot[
    color=SP,
    mark=pentagon*,
    dashed,
    ]
    coordinates {
    (8,0.011)(16, 0.011)(32,0.056)
    };
\addplot[
    color=QFT,
    mark=*,
    dashed,
    ]
    coordinates {
    (8,0.015)(16,0.024)(32,0.032)
    };
\addplot[
    color=FABLE,
    mark=triangle*,
    dashed,
    ]
    coordinates {
    (8,2.678)(16,3.328)(32,16.690)
    };
\addplot[
    color=OAA,
    mark=diamond*,
    dashed,
    ]
    coordinates {
    (8,107.321)(16,133.321)(32,667.801)
    };
\addplot[
    color=red,
    mark=x,
    dashed,
    ]
    coordinates {
    (8,110.040)(16,136.708)(32,684.611)
    };
\end{axis}
\begin{axis}[
        xmin=1, xmax=5,
        xtick={1, 2, 3, 4, 5},
        xticklabels={10, 14, 18, 22, 26},
        ymin=0, ymax=70,
        hide y axis,
        axis x line*=top,
        xlabel={Qubits}
    ]
\end{axis}
\end{tikzpicture}}}{\caption{} \label{Fig: gate sets_d=2}}
    \end{subfloatrow}}
    {\caption{An analysis of the circuit depth scaling for circuit components in the (a) $d=1$ and (b) $d=2$ case for a variety of $n_x$ spatial discretisation steps. Using the scenario described by~\cref{tab: variables}. The results compare two gate sets, solid line for tket and dashed line for \ac{STAR}. The theoretical result is from \cref{thm: QFT}.}
    \label{fig:gate sets d=1and2}}
\end{figure}

\section{Discussion and conclusion}
\label{sec:discuss}
Now that we understand how our algorithm scales for a perfect quantum computer, the final consideration of a resource analysis is
its implementation on an imperfect, noisy quantum computer. While a fault-tolerant resource estimation is beyond the scope of this work, we can comment on resource requirements based on a high-level understanding of quantum error correction. Error correction requires several physical qubits to be encoded into a single logical qubit. The minimum number of physical qubits required to correct all single-qubit errors is $5$, so we use this to estimate a lower bound on the number of physical qubits~\cite{Knill2001BenchmarkingCode}. 
Our minimum scenario of $d=1,\ n_x= 16$ requires $10$ logical qubits, and a meaningful scenario of $d=2,\ n_x = 32$ requires $22$ logical qubits. Therefore, we would require a lower bound of $50$ to $110$ physical qubits. This is comparable to current quantum computers, based on recent works that have documented current qubit capabilities~\cite{Le2025AwesomeComputation, Toshio2025PracticalComputer, Zobrist2025QuantumThreshold}. The limiting factor is the lifetime of these qubits when compared to the number of operations required by our scenarios. To demonstrate this, we consider a superconducting architecture since the \ac{STAR} gate set was designed for a superconducting quantum computer. 
Superconducting devices with at least $150$ qubits are currently capable of around $300$ coherant operations~\cite{Montanez-Barrera2025EvaluatingDepth}.
We note that Ref.~\cite{Tremba2025IsRuntimes} has shown that circuit depth alone is not a good measure of runtime, as it does not take into account the different gate implementation times. For this reason we highlight that this demonstration uses a maximum gate implementation time as a lower bound on current capabilities.
Our two scenarios have circuit depths of $7.51\times 10^3$ and $6.84\times10^5$ respectively from \cref{Fig: gate sets_d=1,Fig: gate sets_d=2}. Both of these exceed the limit of $300$ operations substantially. These results are summarised in~\cref{tab:resource_estimation}.
It is evident that our algorithm cannot be implemented within current qubit lifetimes.

Although quantum error correction can extend qubit lifetimes by a factor of 2-3~\cite{Zobrist2025QuantumThreshold, Dasu2025Order-of-magnitudeCode}, this remains insufficient with current qubits.
Moreover, our estimate represents a lower bound that does not account for the additional operations required for error correction. Incorporating these factors would necessitate consideration of several interdependent variables, such as surface code distance, the physical error rate of the qubits, and the chosen gate set. A comprehensive resource estimation for fault-tolerant implementation would therefore be a valuable direction for future research.

\begin{table}[]
\centering
\begin{tabular}{l c c c c}
\hline
Scenario & Logical Qubits & Physical Qubits & Circuit Depth & \% of Max. \\
\hline
$d = 1$, $n_x = 16$ & $10$ & $\geq 50$ & 7,510 & $2503\%$ \\
$d = 2$, $n_x = 32$ & $22$ & $\geq 110$ & 684,000 & $228,000\%$ \\
\hline
\end{tabular}
\caption{Resource analysis for two representative scenarios. Circuit depth is given as a percentage of the maximum number of operations, $300$~\cite{Montanez-Barrera2025EvaluatingDepth}.} 
\label{tab:resource_estimation}
\end{table}

In conclusion, this paper has presented a practical implementation and resource analysis of a quantum algorithm for solving the \ac{DDE} using diagonalisation via the \ac{QFT} method. We construct and validate the quantum circuit and compare its output against analytical and classical solutions. Our results confirm the algorithm’s correctness and highlight the trade-offs between discretisation accuracy and quantum measurement error.

We then analysed the circuit depth across five gate sets, including both theoretical and hardware-native configurations. We found that scaling with the number of spatial discretisation steps aligns with theoretical predictions. However, scaling with spatial dimension exceeds the theoretical scaling due to the overhead introduced by the \ac{FABLE} subroutine. Among the gate sets, tket and STAR offer the most efficient implementations, with tket showing superior scaling due to its use of multi-qubit gates.

Finally, we lower-bounded the resources required for fault-tolerant implementation of our scenarios. Although the number of logical qubits needed is within reach of current quantum devices, the required circuit depth far exceeds qubit lifetimes, even before accounting for error correction overhead. This underscores the gap between theoretical quantum advantage and the limitations of current hardware, and motivates future work on scalable, error-corrected implementations of quantum algorithms for \acp{PDE}. The existence of error-corrected circuit generation tools is acknowledged; however, compatibility across circuit-building packages is presently insufficient. Subsequent work will involve extending this analysis to an error-corrected circuit for fault-tolerant implementation.

\section*{Acknowledgements}
This work was supported by Fujitsu UK Ltd.

\bibliography{references}

\begin{thebibliography}{41}%
\makeatletter
\providecommand \@ifxundefined [1]{%
 \@ifx{#1\undefined}
}%
\providecommand \@ifnum [1]{%
 \ifnum #1\expandafter \@firstoftwo
 \else \expandafter \@secondoftwo
 \fi
}%
\providecommand \@ifx [1]{%
 \ifx #1\expandafter \@firstoftwo
 \else \expandafter \@secondoftwo
 \fi
}%
\providecommand \natexlab [1]{#1}%
\providecommand \enquote  [1]{``#1''}%
\providecommand \bibnamefont  [1]{#1}%
\providecommand \bibfnamefont [1]{#1}%
\providecommand \citenamefont [1]{#1}%
\providecommand \href@noop [0]{\@secondoftwo}%
\providecommand \href [0]{\begingroup \@sanitize@url \@href}%
\providecommand \@href[1]{\@@startlink{#1}\@@href}%
\providecommand \@@href[1]{\endgroup#1\@@endlink}%
\providecommand \@sanitize@url [0]{\catcode `\\12\catcode `\$12\catcode `\&12\catcode `\#12\catcode `\^12\catcode `\_12\catcode `\%12\relax}%
\providecommand \@@startlink[1]{}%
\providecommand \@@endlink[0]{}%
\providecommand \url  [0]{\begingroup\@sanitize@url \@url }%
\providecommand \@url [1]{\endgroup\@href {#1}{\urlprefix }}%
\providecommand \urlprefix  [0]{URL }%
\providecommand \Eprint [0]{\href }%
\providecommand \doibase [0]{https://doi.org/}%
\providecommand \selectlanguage [0]{\@gobble}%
\providecommand \bibinfo  [0]{\@secondoftwo}%
\providecommand \bibfield  [0]{\@secondoftwo}%
\providecommand \translation [1]{[#1]}%
\providecommand \BibitemOpen [0]{}%
\providecommand \bibitemStop [0]{}%
\providecommand \bibitemNoStop [0]{.\EOS\space}%
\providecommand \EOS [0]{\spacefactor3000\relax}%
\providecommand \BibitemShut  [1]{\csname bibitem#1\endcsname}%
\let\auto@bib@innerbib\@empty
\bibitem [{\citenamefont {Dalzell}\ \emph {et~al.}(2023)\citenamefont {Dalzell}, \citenamefont {McArdle}, \citenamefont {Berta}, \citenamefont {Bienias}, \citenamefont {Chen}, \citenamefont {Gily{\'{e}}n}, \citenamefont {Hann}, \citenamefont {Kastoryano}, \citenamefont {Khabiboulline}, \citenamefont {Kubica}, \citenamefont {Salton}, \citenamefont {Wang},\ and\ \citenamefont {Brand{\~{a}}o}}]{Dalzell2023QuantumComplexities}%
  \BibitemOpen
  \bibfield  {author} {\bibinfo {author} {\bibfnamefont {A.~M.}\ \bibnamefont {Dalzell}}, \bibinfo {author} {\bibfnamefont {S.}~\bibnamefont {McArdle}}, \bibinfo {author} {\bibfnamefont {M.}~\bibnamefont {Berta}}, \bibinfo {author} {\bibfnamefont {P.}~\bibnamefont {Bienias}}, \bibinfo {author} {\bibfnamefont {C.-F.}\ \bibnamefont {Chen}}, \bibinfo {author} {\bibfnamefont {A.}~\bibnamefont {Gily{\'{e}}n}}, \bibinfo {author} {\bibfnamefont {C.~T.}\ \bibnamefont {Hann}}, \bibinfo {author} {\bibfnamefont {M.~J.}\ \bibnamefont {Kastoryano}}, \bibinfo {author} {\bibfnamefont {E.~T.}\ \bibnamefont {Khabiboulline}}, \bibinfo {author} {\bibfnamefont {A.}~\bibnamefont {Kubica}}, \bibinfo {author} {\bibfnamefont {G.}~\bibnamefont {Salton}}, \bibinfo {author} {\bibfnamefont {S.}~\bibnamefont {Wang}},\ and\ \bibinfo {author} {\bibfnamefont {F.~G. S.~L.}\ \bibnamefont {Brand{\~{a}}o}},\ }\bibfield  {title} {\bibinfo {title} {{Quantum algorithms: A survey of applications and end-to-end complexities}},\ }\href
  {https://arxiv.org/pdf/2310.03011} {\bibfield  {journal} {\bibinfo  {journal} {ArXiv}\ } (\bibinfo {year} {2023})},\ \Eprint {https://arxiv.org/abs/2310.03011} {arXiv:2310.03011 [quant-ph]} \BibitemShut {NoStop}%
\bibitem [{\citenamefont {Gheorghiu}\ and\ \citenamefont {Mosca}(2025)}]{Gheorghiu2025QuantumAlgorithms}%
  \BibitemOpen
  \bibfield  {author} {\bibinfo {author} {\bibfnamefont {V.}~\bibnamefont {Gheorghiu}}\ and\ \bibinfo {author} {\bibfnamefont {M.}~\bibnamefont {Mosca}},\ }\bibfield  {title} {\bibinfo {title} {{Quantum resource estimation for large scale quantum algorithms}},\ }\href {https://doi.org/10.1016/J.FUTURE.2024.107480} {\bibfield  {journal} {\bibinfo  {journal} {Future Generation Computer Systems}\ }\textbf {\bibinfo {volume} {162}},\ \bibinfo {pages} {107480} (\bibinfo {year} {2025})}\BibitemShut {NoStop}%
\bibitem [{\citenamefont {Gidney}\ and\ \citenamefont {Eker{\aa}}(2021)}]{Gidney2021HowQubits}%
  \BibitemOpen
  \bibfield  {author} {\bibinfo {author} {\bibfnamefont {C.}~\bibnamefont {Gidney}}\ and\ \bibinfo {author} {\bibfnamefont {M.}~\bibnamefont {Eker{\aa}}},\ }\bibfield  {title} {\bibinfo {title} {{How to factor 2048 bit RSA integers in 8 hours using 20 million noisy qubits}},\ }\href {https://doi.org/10.22331/q-2021-04-15-433} {\bibfield  {journal} {\bibinfo  {journal} {Quantum}\ }\textbf {\bibinfo {volume} {5}},\ \bibinfo {pages} {1} (\bibinfo {year} {2021})}\BibitemShut {NoStop}%
\bibitem [{\citenamefont {Yamaguchi}\ \emph {et~al.}(2023)\citenamefont {Yamaguchi}, \citenamefont {Yamazaki}, \citenamefont {Tabuchi}, \citenamefont {Honda}, \citenamefont {Izu},\ and\ \citenamefont {Kunihiro}}]{Yamaguchi2023ExperimentsSimulator}%
  \BibitemOpen
  \bibfield  {author} {\bibinfo {author} {\bibfnamefont {J.}~\bibnamefont {Yamaguchi}}, \bibinfo {author} {\bibfnamefont {M.}~\bibnamefont {Yamazaki}}, \bibinfo {author} {\bibfnamefont {A.}~\bibnamefont {Tabuchi}}, \bibinfo {author} {\bibfnamefont {T.}~\bibnamefont {Honda}}, \bibinfo {author} {\bibfnamefont {T.}~\bibnamefont {Izu}},\ and\ \bibinfo {author} {\bibfnamefont {N.}~\bibnamefont {Kunihiro}},\ }\bibfield  {title} {\bibinfo {title} {{Experiments and Resource Analysis of Shor’s Factorization Using a Quantum Simulator}},\ }\href {https://doi.org/https://doi.org/10.1007/978-981-97-1235-9_7} {\bibfield  {journal} {\bibinfo  {journal} {Lecture Notes in Computer Science (including subseries Lecture Notes in Artificial Intelligence and Lecture Notes in Bioinformatics)}\ }\textbf {\bibinfo {volume} {14561 LNCS}},\ \bibinfo {pages} {119} (\bibinfo {year} {2023})}\BibitemShut {NoStop}%
\bibitem [{\citenamefont {Katabarwa}\ \emph {et~al.}(2024)\citenamefont {Katabarwa}, \citenamefont {Gratsea}, \citenamefont {Caesura},\ and\ \citenamefont {Johnson}}]{Katabarwa2024EarlyComputing}%
  \BibitemOpen
  \bibfield  {author} {\bibinfo {author} {\bibfnamefont {A.}~\bibnamefont {Katabarwa}}, \bibinfo {author} {\bibfnamefont {K.}~\bibnamefont {Gratsea}}, \bibinfo {author} {\bibfnamefont {A.}~\bibnamefont {Caesura}},\ and\ \bibinfo {author} {\bibfnamefont {P.~D.}\ \bibnamefont {Johnson}},\ }\bibfield  {title} {\bibinfo {title} {{Early Fault-Tolerant Quantum Computing}},\ }\href {https://doi.org/10.1103/PRXQUANTUM.5.020101} {\bibfield  {journal} {\bibinfo  {journal} {PRX Quantum}\ }\textbf {\bibinfo {volume} {5}},\ \bibinfo {pages} {020101} (\bibinfo {year} {2024})}\BibitemShut {NoStop}%
\bibitem [{\citenamefont {Kiss}\ \emph {et~al.}(2025)\citenamefont {Kiss}, \citenamefont {Azad}, \citenamefont {Requena}, \citenamefont {Roggero}, \citenamefont {Wakeham},\ and\ \citenamefont {Arrazola}}]{Kiss2025EarlyEstimation}%
  \BibitemOpen
  \bibfield  {author} {\bibinfo {author} {\bibfnamefont {O.}~\bibnamefont {Kiss}}, \bibinfo {author} {\bibfnamefont {U.}~\bibnamefont {Azad}}, \bibinfo {author} {\bibfnamefont {B.}~\bibnamefont {Requena}}, \bibinfo {author} {\bibfnamefont {A.}~\bibnamefont {Roggero}}, \bibinfo {author} {\bibfnamefont {D.}~\bibnamefont {Wakeham}},\ and\ \bibinfo {author} {\bibfnamefont {J.~M.}\ \bibnamefont {Arrazola}},\ }\bibfield  {title} {\bibinfo {title} {{Early Fault-Tolerant Quantum Algorithms in Practice: Application to Ground-State Energy Estimation}},\ }\href {https://doi.org/10.22331/q-2025-04-01-1682} {\bibfield  {journal} {\bibinfo  {journal} {Quantum}\ }\textbf {\bibinfo {volume} {9}},\ \bibinfo {pages} {1682} (\bibinfo {year} {2025})}\BibitemShut {NoStop}%
\bibitem [{\citenamefont {Arenas-L{\'{o}}pez}\ and\ \citenamefont {Badaoui}(2020)}]{Arenas-Lopez2020AOutput}%
  \BibitemOpen
  \bibfield  {author} {\bibinfo {author} {\bibfnamefont {J.~P.}\ \bibnamefont {Arenas-L{\'{o}}pez}}\ and\ \bibinfo {author} {\bibfnamefont {M.}~\bibnamefont {Badaoui}},\ }\bibfield  {title} {\bibinfo {title} {{A Fokker–Planck equation based approach for modelling wind speed and its power output}},\ }\href {https://doi.org/10.1016/j.enconman.2020.113152} {\bibfield  {journal} {\bibinfo  {journal} {Energy Conversion and Management}\ }\textbf {\bibinfo {volume} {222}},\ \bibinfo {pages} {113152} (\bibinfo {year} {2020})}\BibitemShut {NoStop}%
\bibitem [{\citenamefont {Montanaro}(2015)}]{Montanaro2015QuantumMethods}%
  \BibitemOpen
  \bibfield  {author} {\bibinfo {author} {\bibfnamefont {A.}~\bibnamefont {Montanaro}},\ }\bibfield  {title} {\bibinfo {title} {{Quantum speedup of Monte Carlo methods}},\ }\href {https://doi.org/10.1098/rspa.2015.0301} {\bibfield  {journal} {\bibinfo  {journal} {Proceedings of the Royal Society A: Mathematical, Physical and Engineering Sciences}\ }\textbf {\bibinfo {volume} {471}},\ \bibinfo {pages} {20150301} (\bibinfo {year} {2015})}\BibitemShut {NoStop}%
\bibitem [{\citenamefont {Andersen}\ \emph {et~al.}(2001)\citenamefont {Andersen}, \citenamefont {Bollerslev}, \citenamefont {Diebold},\ and\ \citenamefont {Ebens}}]{Andersen2001TheVolatility}%
  \BibitemOpen
  \bibfield  {author} {\bibinfo {author} {\bibfnamefont {T.~G.}\ \bibnamefont {Andersen}}, \bibinfo {author} {\bibfnamefont {T.}~\bibnamefont {Bollerslev}}, \bibinfo {author} {\bibfnamefont {F.~X.}\ \bibnamefont {Diebold}},\ and\ \bibinfo {author} {\bibfnamefont {H.}~\bibnamefont {Ebens}},\ }\bibfield  {title} {\bibinfo {title} {{The distribution of realized stock return volatility}},\ }\href {https://doi.org/10.1016/S0304-405X(01)00055-1} {\bibfield  {journal} {\bibinfo  {journal} {Journal of Financial Economics}\ }\textbf {\bibinfo {volume} {61}},\ \bibinfo {pages} {43} (\bibinfo {year} {2001})}\BibitemShut {NoStop}%
\bibitem [{\citenamefont {Wang}\ \emph {et~al.}(2023)\citenamefont {Wang}, \citenamefont {Xiao},\ and\ \citenamefont {Yu}}]{Wang2023ModelingProcess}%
  \BibitemOpen
  \bibfield  {author} {\bibinfo {author} {\bibfnamefont {X.}~\bibnamefont {Wang}}, \bibinfo {author} {\bibfnamefont {W.}~\bibnamefont {Xiao}},\ and\ \bibinfo {author} {\bibfnamefont {J.}~\bibnamefont {Yu}},\ }\bibfield  {title} {\bibinfo {title} {{Modeling and forecasting realized volatility with the fractional Ornstein–Uhlenbeck process}},\ }\href {https://doi.org/10.1016/J.JECONOM.2021.08.001} {\bibfield  {journal} {\bibinfo  {journal} {Journal of Econometrics}\ }\textbf {\bibinfo {volume} {232}},\ \bibinfo {pages} {389} (\bibinfo {year} {2023})}\BibitemShut {NoStop}%
\bibitem [{\citenamefont {Novikau}\ and\ \citenamefont {Joseph}(2024)}]{Novikau2024QuantumSystems}%
  \BibitemOpen
  \bibfield  {author} {\bibinfo {author} {\bibfnamefont {I.}~\bibnamefont {Novikau}}\ and\ \bibinfo {author} {\bibfnamefont {I.}~\bibnamefont {Joseph}},\ }\bibfield  {title} {\bibinfo {title} {{Quantum algorithm for the advection-diffusion equation and the Koopman-von Neumann approach to nonlinear dynamical systems}},\ }\href {http://arxiv.org/abs/2410.03985} {\bibfield  {journal} {\bibinfo  {journal} {ArXiv}\ } (\bibinfo {year} {2024})},\ \Eprint {https://arxiv.org/abs/2410.03985} {arXiv:2410.03985 [quant-ph]} \BibitemShut {NoStop}%
\bibitem [{\citenamefont {Lubasch}\ \emph {et~al.}(2025)\citenamefont {Lubasch}, \citenamefont {Kikuchi}, \citenamefont {Wright},\ and\ \citenamefont {Keever}}]{Lubasch2025QuantumSpace}%
  \BibitemOpen
  \bibfield  {author} {\bibinfo {author} {\bibfnamefont {M.}~\bibnamefont {Lubasch}}, \bibinfo {author} {\bibfnamefont {Y.}~\bibnamefont {Kikuchi}}, \bibinfo {author} {\bibfnamefont {L.}~\bibnamefont {Wright}},\ and\ \bibinfo {author} {\bibfnamefont {C.~M.}\ \bibnamefont {Keever}},\ }\bibfield  {title} {\bibinfo {title} {{Quantum circuits for partial differential equations in Fourier space}},\ }\href {https://arxiv.org/pdf/2505.16895} {\bibfield  {journal} {\bibinfo  {journal} {ArXiv}\ } (\bibinfo {year} {2025})},\ \Eprint {https://arxiv.org/abs/2505.16895} {arXiv:2505.16895 [quant-ph]} \BibitemShut {NoStop}%
\bibitem [{\citenamefont {Brearley}\ and\ \citenamefont {Laizet}(2024)}]{Brearley2024QuantumSimulation}%
  \BibitemOpen
  \bibfield  {author} {\bibinfo {author} {\bibfnamefont {P.}~\bibnamefont {Brearley}}\ and\ \bibinfo {author} {\bibfnamefont {S.}~\bibnamefont {Laizet}},\ }\bibfield  {title} {\bibinfo {title} {{Quantum algorithm for solving the advection equation using Hamiltonian simulation}},\ }\href {https://doi.org/10.1103/PHYSREVA.110.012430} {\bibfield  {journal} {\bibinfo  {journal} {Physical Review A}\ }\textbf {\bibinfo {volume} {110}},\ \bibinfo {pages} {012430} (\bibinfo {year} {2024})}\BibitemShut {NoStop}%
\bibitem [{\citenamefont {Devereux}\ and\ \citenamefont {Datta}(2025)}]{Devereux2025QuantumEquation}%
  \BibitemOpen
  \bibfield  {author} {\bibinfo {author} {\bibfnamefont {E.}~\bibnamefont {Devereux}}\ and\ \bibinfo {author} {\bibfnamefont {A.}~\bibnamefont {Datta}},\ }\bibfield  {title} {\bibinfo {title} {{Quantum algorithms for solving a drift-diffusion equation}},\ }\href {https://arxiv.org/pdf/2505.21221} {\bibfield  {journal} {\bibinfo  {journal} {ArXiv}\ } (\bibinfo {year} {2025})},\ \Eprint {https://arxiv.org/abs/2505.21221} {arXiv:2505.21221 [quant-ph]} \BibitemShut {NoStop}%
\bibitem [{\citenamefont {Novikau}\ and\ \citenamefont {Joseph}(2025)}]{NovikauExplicitEquation}%
  \BibitemOpen
  \bibfield  {author} {\bibinfo {author} {\bibfnamefont {I.}~\bibnamefont {Novikau}}\ and\ \bibinfo {author} {\bibfnamefont {I.}~\bibnamefont {Joseph}},\ }\bibfield  {title} {\bibinfo {title} {{Explicit near-optimal quantum algorithm for solving the advection-diffusion equation}},\ }\href {https://arxiv.org/abs/2501.11146v1} {\bibfield  {journal} {\bibinfo  {journal} {ArXiv}\ } (\bibinfo {year} {2025})},\ \Eprint {https://arxiv.org/abs/2501.11146v1} {arXiv:2501.11146v1 [quant-ph]} \BibitemShut {NoStop}%
\bibitem [{\citenamefont {Over}\ \emph {et~al.}(2025)\citenamefont {Over}, \citenamefont {Bengoechea}, \citenamefont {Brearley}, \citenamefont {Laizet},\ and\ \citenamefont {Rung}}]{Over2025QuantumOperator}%
  \BibitemOpen
  \bibfield  {author} {\bibinfo {author} {\bibfnamefont {P.}~\bibnamefont {Over}}, \bibinfo {author} {\bibfnamefont {S.}~\bibnamefont {Bengoechea}}, \bibinfo {author} {\bibfnamefont {P.}~\bibnamefont {Brearley}}, \bibinfo {author} {\bibfnamefont {S.}~\bibnamefont {Laizet}},\ and\ \bibinfo {author} {\bibfnamefont {T.}~\bibnamefont {Rung}},\ }\bibfield  {title} {\bibinfo {title} {{Quantum algorithm for the advection-diffusion equation by direct block encoding of the time-marching operator}},\ }\href {https://doi.org/10.1103/d8hb-fv93} {\bibfield  {journal} {\bibinfo  {journal} {Physical Review A}\ }\textbf {\bibinfo {volume} {112}},\ \bibinfo {pages} {10401} (\bibinfo {year} {2025})}\BibitemShut {NoStop}%
\bibitem [{\citenamefont {Sivarajah}\ \emph {et~al.}(2020)\citenamefont {Sivarajah}, \citenamefont {Dilkes}, \citenamefont {Cowtan}, \citenamefont {Simmons}, \citenamefont {Edgington},\ and\ \citenamefont {Duncan}}]{Sivarajah2020Tket:Devices}%
  \BibitemOpen
  \bibfield  {author} {\bibinfo {author} {\bibfnamefont {S.}~\bibnamefont {Sivarajah}}, \bibinfo {author} {\bibfnamefont {S.}~\bibnamefont {Dilkes}}, \bibinfo {author} {\bibfnamefont {A.}~\bibnamefont {Cowtan}}, \bibinfo {author} {\bibfnamefont {W.}~\bibnamefont {Simmons}}, \bibinfo {author} {\bibfnamefont {A.}~\bibnamefont {Edgington}},\ and\ \bibinfo {author} {\bibfnamefont {R.}~\bibnamefont {Duncan}},\ }\bibfield  {title} {\bibinfo {title} {{t|ket⟩: a retargetable compiler for NISQ devices}},\ }\href {https://doi.org/10.1088/2058-9565/AB8E92} {\bibfield  {journal} {\bibinfo  {journal} {Quantum Science and Technology}\ }\textbf {\bibinfo {volume} {6}},\ \bibinfo {pages} {014003} (\bibinfo {year} {2020})}\BibitemShut {NoStop}%
\bibitem [{\citenamefont {IBM}()}]{ProcessorDocumentation}%
  \BibitemOpen
  \bibfield  {author} {\bibinfo {author} {\bibnamefont {IBM}},\ }\href {https://quantum.cloud.ibm.com/docs/en/guides/processor-types} {\bibinfo {title} {{Processor types | IBM Quantum Documentation}}},\ \bibinfo {note} {accessed: 11/08/2025}\BibitemShut {NoStop}%
\bibitem [{\citenamefont {Montanez-Barrera}\ \emph {et~al.}(2025)\citenamefont {Montanez-Barrera}, \citenamefont {Michielsen},\ and\ \citenamefont {Neira}}]{Montanez-Barrera2025EvaluatingDepth}%
  \BibitemOpen
  \bibfield  {author} {\bibinfo {author} {\bibfnamefont {J.~A.}\ \bibnamefont {Montanez-Barrera}}, \bibinfo {author} {\bibfnamefont {K.}~\bibnamefont {Michielsen}},\ and\ \bibinfo {author} {\bibfnamefont {D.~E.~B.}\ \bibnamefont {Neira}},\ }\bibfield  {title} {\bibinfo {title} {{Evaluating the performance of quantum processing units at large width and depth}},\ }\href {https://arxiv.org/pdf/2502.06471} {\bibfield  {journal} {\bibinfo  {journal} {ArXiv}\ } (\bibinfo {year} {2025})},\ \Eprint {https://arxiv.org/abs/2502.06471} {arXiv:2502.06471 [quant-ph]} \BibitemShut {NoStop}%
\bibitem [{\citenamefont {Toshio}\ \emph {et~al.}(2025)\citenamefont {Toshio}, \citenamefont {Akahoshi}, \citenamefont {Fujisaki}, \citenamefont {Oshima}, \citenamefont {Sato},\ and\ \citenamefont {Fujii}}]{Toshio2025PracticalComputer}%
  \BibitemOpen
  \bibfield  {author} {\bibinfo {author} {\bibfnamefont {R.}~\bibnamefont {Toshio}}, \bibinfo {author} {\bibfnamefont {Y.}~\bibnamefont {Akahoshi}}, \bibinfo {author} {\bibfnamefont {J.}~\bibnamefont {Fujisaki}}, \bibinfo {author} {\bibfnamefont {H.}~\bibnamefont {Oshima}}, \bibinfo {author} {\bibfnamefont {S.}~\bibnamefont {Sato}},\ and\ \bibinfo {author} {\bibfnamefont {K.}~\bibnamefont {Fujii}},\ }\bibfield  {title} {\bibinfo {title} {{Practical Quantum Advantage on Partially Fault-Tolerant Quantum Computer}},\ }\href {https://doi.org/10.1103/PhysRevX.15.021057} {\bibfield  {journal} {\bibinfo  {journal} {Physical Review X}\ }\textbf {\bibinfo {volume} {15}},\ \bibinfo {pages} {021057} (\bibinfo {year} {2025})}\BibitemShut {NoStop}%
\bibitem [{\citenamefont {IonQ}()}]{NativeDocumentation}%
  \BibitemOpen
  \bibfield  {author} {\bibinfo {author} {\bibnamefont {IonQ}},\ }\href {https://docs.ionq.com/guides/getting-started-with-native-gates} {\bibinfo {title} {{Native Gates - IonQ Quantum Cloud Documentation}}},\ \bibinfo {note} {accessed: 18/07/2025}\BibitemShut {NoStop}%
\bibitem [{\citenamefont {Camps}\ and\ \citenamefont {Van~Beeumen}(2022)}]{Camps2022FABLE:Block-Encodings}%
  \BibitemOpen
  \bibfield  {author} {\bibinfo {author} {\bibfnamefont {D.}~\bibnamefont {Camps}}\ and\ \bibinfo {author} {\bibfnamefont {R.}~\bibnamefont {Van~Beeumen}},\ }\bibfield  {title} {\bibinfo {title} {{FABLE: Fast Approximate Quantum Circuits for Block-Encodings}},\ }\href {https://doi.org/10.1109/QCE53715.2022.00029} {\bibfield  {journal} {\bibinfo  {journal} {Proceedings - 2022 IEEE International Conference on Quantum Computing and Engineering, QCE 2022}\ ,\ \bibinfo {pages} {104}} (\bibinfo {year} {2022})}\BibitemShut {NoStop}%
\bibitem [{\citenamefont {Risken}(1996)}]{Risken1996Fokker-PlanckEquation}%
  \BibitemOpen
  \bibfield  {author} {\bibinfo {author} {\bibfnamefont {H.}~\bibnamefont {Risken}},\ }\href {https://doi.org/10.1007/978-3-642-61544-3\_4} {\emph {\bibinfo {title} {{Fokker-Planck Equation}}}},\ \bibinfo {edition} {2nd}\ ed.,\ edited by\ \bibinfo {editor} {\bibfnamefont {H.}~\bibnamefont {Haken}},\ Vol.\ \bibinfo {volume} {SSSYN volume 18}\ (\bibinfo  {publisher} {Springer},\ \bibinfo {year} {1996})\BibitemShut {NoStop}%
\bibitem [{\citenamefont {Araujo}\ \emph {et~al.}(2023)\citenamefont {Araujo}, \citenamefont {Ara{\'{u}}jo}, \citenamefont {da~Silva}, \citenamefont {Blank},\ and\ \citenamefont {da~Silva}}]{Araujo2023QuantumLibrary}%
  \BibitemOpen
  \bibfield  {author} {\bibinfo {author} {\bibfnamefont {I.~F.}\ \bibnamefont {Araujo}}, \bibinfo {author} {\bibfnamefont {I.~C.~S.}\ \bibnamefont {Ara{\'{u}}jo}}, \bibinfo {author} {\bibfnamefont {L.~D.}\ \bibnamefont {da~Silva}}, \bibinfo {author} {\bibfnamefont {C.}~\bibnamefont {Blank}},\ and\ \bibinfo {author} {\bibfnamefont {A.~J.}\ \bibnamefont {da~Silva}},\ }\href {https://github.com/qclib/qclib} {\bibinfo {title} {{Quantum Computing Library}}} (\bibinfo {year} {2023})\BibitemShut {NoStop}%
\bibitem [{\citenamefont {Araujo}\ \emph {et~al.}(2024)\citenamefont {Araujo}, \citenamefont {Blank}, \citenamefont {Araujo},\ and\ \citenamefont {Da~Silva}}]{Araujo2024Low-RankPreparation}%
  \BibitemOpen
  \bibfield  {author} {\bibinfo {author} {\bibfnamefont {I.~F.}\ \bibnamefont {Araujo}}, \bibinfo {author} {\bibfnamefont {C.}~\bibnamefont {Blank}}, \bibinfo {author} {\bibfnamefont {I.~C.}\ \bibnamefont {Araujo}},\ and\ \bibinfo {author} {\bibfnamefont {A.~J.}\ \bibnamefont {Da~Silva}},\ }\bibfield  {title} {\bibinfo {title} {{Low-Rank Quantum State Preparation}},\ }\href {https://doi.org/10.1109/TCAD.2023.3297972} {\bibfield  {journal} {\bibinfo  {journal} {IEEE Transactions on Computer-Aided Design of Integrated Circuits and Systems}\ }\textbf {\bibinfo {volume} {43}},\ \bibinfo {pages} {161} (\bibinfo {year} {2024})}\BibitemShut {NoStop}%
\bibitem [{\citenamefont {Nielson}\ and\ \citenamefont {Chuang}(2010)}]{Nielson2010QuantumInformation}%
  \BibitemOpen
  \bibfield  {author} {\bibinfo {author} {\bibfnamefont {M.}~\bibnamefont {Nielson}}\ and\ \bibinfo {author} {\bibfnamefont {I.}~\bibnamefont {Chuang}},\ }\href {https://books.google.co.uk/books?id=-s4DEy7o-a0C} {\emph {\bibinfo {title} {{Quantum Computation and Quantum Information}}}},\ \bibinfo {edition} {10th}\ ed.\ (\bibinfo  {publisher} {Cambridge University Press},\ \bibinfo {year} {2010})\BibitemShut {NoStop}%
\bibitem [{\citenamefont {Pfeffer}(2023)}]{PfefferMultidimensionalTransformation}%
  \BibitemOpen
  \bibfield  {author} {\bibinfo {author} {\bibfnamefont {P.}~\bibnamefont {Pfeffer}},\ }\bibfield  {title} {\bibinfo {title} {{Multidimensional Quantum Fourier Transformation}},\ }\href {https://arxiv.org/abs/2301.13835v1} {\bibfield  {journal} {\bibinfo  {journal} {ArXiv}\ } (\bibinfo {year} {2023})},\ \Eprint {https://arxiv.org/abs/2301.13835v1} {arXiv:2301.13835v1 [quant-ph]} \BibitemShut {NoStop}%
\bibitem [{\citenamefont {Musk}(2020)}]{Musk2020AComputations}%
  \BibitemOpen
  \bibfield  {author} {\bibinfo {author} {\bibfnamefont {D.~R.}\ \bibnamefont {Musk}},\ }\bibfield  {title} {\bibinfo {title} {{A Comparison of Quantum and Traditional Fourier Transform Computations}},\ }\href {https://doi.org/10.1109/MCSE.2020.3023979} {\bibfield  {journal} {\bibinfo  {journal} {Computing in Science {\&} Engineering}\ }\textbf {\bibinfo {volume} {22}},\ \bibinfo {pages} {103} (\bibinfo {year} {2020})}\BibitemShut {NoStop}%
\bibitem [{\citenamefont {Nam}\ \emph {et~al.}(2020)\citenamefont {Nam}, \citenamefont {Su},\ and\ \citenamefont {Maslov}}]{Nam2020ApproximateGates}%
  \BibitemOpen
  \bibfield  {author} {\bibinfo {author} {\bibfnamefont {Y.}~\bibnamefont {Nam}}, \bibinfo {author} {\bibfnamefont {Y.}~\bibnamefont {Su}},\ and\ \bibinfo {author} {\bibfnamefont {D.}~\bibnamefont {Maslov}},\ }\bibfield  {title} {\bibinfo {title} {{Approximate quantum Fourier transform with O(n log(n)) T gates}},\ }\href {https://doi.org/10.1038/S41534-020-0257-5} {\bibfield  {journal} {\bibinfo  {journal} {npj Quantum Information}\ }\textbf {\bibinfo {volume} {6}},\ \bibinfo {pages} {1} (\bibinfo {year} {2020})}\BibitemShut {NoStop}%
\bibitem [{\citenamefont {Brassard}\ \emph {et~al.}(2002)\citenamefont {Brassard}, \citenamefont {H{\o}yer}, \citenamefont {Mosca},\ and\ \citenamefont {Tapp}}]{Brassard2002QuantumEstimation}%
  \BibitemOpen
  \bibfield  {author} {\bibinfo {author} {\bibfnamefont {G.}~\bibnamefont {Brassard}}, \bibinfo {author} {\bibfnamefont {P.}~\bibnamefont {H{\o}yer}}, \bibinfo {author} {\bibfnamefont {M.}~\bibnamefont {Mosca}},\ and\ \bibinfo {author} {\bibfnamefont {A.}~\bibnamefont {Tapp}},\ }\bibfield  {title} {\bibinfo {title} {{Quantum amplitude amplification and estimation}},\ }in\ \href {https://doi.org/10.1090/conm/305/05215} {\emph {\bibinfo {booktitle} {Quantum Computation and Information}}},\ Vol.\ \bibinfo {volume} {305}\ (\bibinfo  {publisher} {AMS},\ \bibinfo {year} {2002})\ pp.\ \bibinfo {pages} {53--74}\BibitemShut {NoStop}%
\bibitem [{\citenamefont {Berry}\ \emph {et~al.}(2014)\citenamefont {Berry}, \citenamefont {Childs}, \citenamefont {Cleve}, \citenamefont {Kothari},\ and\ \citenamefont {Somma}}]{BerryExponential}%
  \BibitemOpen
  \bibfield  {author} {\bibinfo {author} {\bibfnamefont {D.~W.}\ \bibnamefont {Berry}}, \bibinfo {author} {\bibfnamefont {A.~M.}\ \bibnamefont {Childs}}, \bibinfo {author} {\bibfnamefont {R.}~\bibnamefont {Cleve}}, \bibinfo {author} {\bibfnamefont {R.}~\bibnamefont {Kothari}},\ and\ \bibinfo {author} {\bibfnamefont {R.~D.}\ \bibnamefont {Somma}},\ }\bibfield  {title} {\bibinfo {title} {{Exponential improvement in precision for simulating sparse Hamiltonians}},\ }in\ \href {https://doi.org/10.1145/2591796.2591854} {\emph {\bibinfo {booktitle} {Proceedings of the 46th ACM Symposium on Theory of Computing (STOC 2014)}}}\ (\bibinfo {year} {2014})\BibitemShut {NoStop}%
\bibitem [{\citenamefont {Van~Apeldoorn}(2021)}]{VanApeldoorn2021QuantumEstimation}%
  \BibitemOpen
  \bibfield  {author} {\bibinfo {author} {\bibfnamefont {J.}~\bibnamefont {Van~Apeldoorn}},\ }\bibfield  {title} {\bibinfo {title} {{Quantum Probability Oracles {\&} Multidimensional Amplitude Estimation}},\ }in\ \href {https://doi.org/10.4230/LIPIcs.TQC.2021.9} {\emph {\bibinfo {booktitle} {16th Conference on the Theory of Quantum Computation, Communication and Cryptography (TQC 2021)}}},\ Vol.\ \bibinfo {volume} {197}\ (\bibinfo {year} {2021})\ pp.\ \bibinfo {pages} {1--9}\BibitemShut {NoStop}%
\bibitem [{\citenamefont {Suzuki}\ \emph {et~al.}(2021)\citenamefont {Suzuki}, \citenamefont {Kawase}, \citenamefont {Masumura}, \citenamefont {Hiraga}, \citenamefont {Nakadai}, \citenamefont {Chen}, \citenamefont {Nakanishi}, \citenamefont {Mitarai}, \citenamefont {Imai}, \citenamefont {Tamiya}, \citenamefont {Yamamoto}, \citenamefont {Yan}, \citenamefont {Kawakubo}, \citenamefont {Nakagawa}, \citenamefont {Ibe}, \citenamefont {Zhang}, \citenamefont {Yamashita}, \citenamefont {Yoshimura}, \citenamefont {Hayashi},\ and\ \citenamefont {Fujii}}]{Suzuki2021Qulacs:Purpose}%
  \BibitemOpen
  \bibfield  {author} {\bibinfo {author} {\bibfnamefont {Y.}~\bibnamefont {Suzuki}}, \bibinfo {author} {\bibfnamefont {Y.}~\bibnamefont {Kawase}}, \bibinfo {author} {\bibfnamefont {Y.}~\bibnamefont {Masumura}}, \bibinfo {author} {\bibfnamefont {Y.}~\bibnamefont {Hiraga}}, \bibinfo {author} {\bibfnamefont {M.}~\bibnamefont {Nakadai}}, \bibinfo {author} {\bibfnamefont {J.}~\bibnamefont {Chen}}, \bibinfo {author} {\bibfnamefont {K.~M.}\ \bibnamefont {Nakanishi}}, \bibinfo {author} {\bibfnamefont {K.}~\bibnamefont {Mitarai}}, \bibinfo {author} {\bibfnamefont {R.}~\bibnamefont {Imai}}, \bibinfo {author} {\bibfnamefont {S.}~\bibnamefont {Tamiya}}, \bibinfo {author} {\bibfnamefont {T.}~\bibnamefont {Yamamoto}}, \bibinfo {author} {\bibfnamefont {T.}~\bibnamefont {Yan}}, \bibinfo {author} {\bibfnamefont {T.}~\bibnamefont {Kawakubo}}, \bibinfo {author} {\bibfnamefont {Y.~O.}\ \bibnamefont {Nakagawa}}, \bibinfo {author} {\bibfnamefont {Y.}~\bibnamefont {Ibe}}, \bibinfo {author} {\bibfnamefont {Y.}~\bibnamefont {Zhang}},
  \bibinfo {author} {\bibfnamefont {H.}~\bibnamefont {Yamashita}}, \bibinfo {author} {\bibfnamefont {H.}~\bibnamefont {Yoshimura}}, \bibinfo {author} {\bibfnamefont {A.}~\bibnamefont {Hayashi}},\ and\ \bibinfo {author} {\bibfnamefont {K.}~\bibnamefont {Fujii}},\ }\bibfield  {title} {\bibinfo {title} {{Qulacs: a fast and versatile quantum circuit simulator for research purpose}},\ }\href {https://doi.org/10.22331/q-2021-10-06-559} {\bibfield  {journal} {\bibinfo  {journal} {Quantum}\ }\textbf {\bibinfo {volume} {5}},\ \bibinfo {pages} {559} (\bibinfo {year} {2021})}\BibitemShut {NoStop}%
\bibitem [{\citenamefont {Linden}\ \emph {et~al.}(2022)\citenamefont {Linden}, \citenamefont {Montanaro},\ and\ \citenamefont {Shao}}]{Linden2022QuantumEquation}%
  \BibitemOpen
  \bibfield  {author} {\bibinfo {author} {\bibfnamefont {N.}~\bibnamefont {Linden}}, \bibinfo {author} {\bibfnamefont {A.}~\bibnamefont {Montanaro}},\ and\ \bibinfo {author} {\bibfnamefont {C.}~\bibnamefont {Shao}},\ }\bibfield  {title} {\bibinfo {title} {{Quantum vs. Classical Algorithms for Solving the Heat Equation}},\ }\href {https://doi.org/10.1007/S00220-022-04442-6} {\bibfield  {journal} {\bibinfo  {journal} {Communications in Mathematical Physics}\ }\textbf {\bibinfo {volume} {395}},\ \bibinfo {pages} {601} (\bibinfo {year} {2022})}\BibitemShut {NoStop}%
\bibitem [{\citenamefont {Plesch}\ and\ \citenamefont {Andˇand{\v{c}}aslav~Brukner}(2011)}]{Plesch2011Quantum-stateDecompositions}%
  \BibitemOpen
  \bibfield  {author} {\bibinfo {author} {\bibfnamefont {M.}~\bibnamefont {Plesch}}\ and\ \bibinfo {author} {\bibfnamefont {A.}~\bibnamefont {Andˇand{\v{c}}aslav~Brukner}},\ }\bibfield  {title} {\bibinfo {title} {{Quantum-state preparation with universal gate decompositions}},\ }\href {https://doi.org/10.1103/PhysRevA.83.032302} {\bibfield  {journal} {\bibinfo  {journal} {Physical Review A}\ }\textbf {\bibinfo {volume} {83}},\ \bibinfo {pages} {32302} (\bibinfo {year} {2011})}\BibitemShut {NoStop}%
\bibitem [{\citenamefont {Arrazola}\ \emph {et~al.}(2019)\citenamefont {Arrazola}, \citenamefont {Kalajdzievski}, \citenamefont {Weedbrook},\ and\ \citenamefont {Lloyd}}]{Arrazola2019QuantumEquations}%
  \BibitemOpen
  \bibfield  {author} {\bibinfo {author} {\bibfnamefont {J.~M.}\ \bibnamefont {Arrazola}}, \bibinfo {author} {\bibfnamefont {T.}~\bibnamefont {Kalajdzievski}}, \bibinfo {author} {\bibfnamefont {C.}~\bibnamefont {Weedbrook}},\ and\ \bibinfo {author} {\bibfnamefont {S.}~\bibnamefont {Lloyd}},\ }\bibfield  {title} {\bibinfo {title} {{Quantum algorithm for nonhomogeneous linear partial differential equations}},\ }\href {https://doi.org/10.1103/PhysRevA.100.032306} {\bibfield  {journal} {\bibinfo  {journal} {Physical Review A}\ }\textbf {\bibinfo {volume} {100}},\ \bibinfo {pages} {32306} (\bibinfo {year} {2019})}\BibitemShut {NoStop}%
\bibitem [{\citenamefont {Knill}\ \emph {et~al.}(2001)\citenamefont {Knill}, \citenamefont {Laflamme}, \citenamefont {Martinez},\ and\ \citenamefont {Negrevergne}}]{Knill2001BenchmarkingCode}%
  \BibitemOpen
  \bibfield  {author} {\bibinfo {author} {\bibfnamefont {E.}~\bibnamefont {Knill}}, \bibinfo {author} {\bibfnamefont {R.}~\bibnamefont {Laflamme}}, \bibinfo {author} {\bibfnamefont {R.}~\bibnamefont {Martinez}},\ and\ \bibinfo {author} {\bibfnamefont {C.}~\bibnamefont {Negrevergne}},\ }\bibfield  {title} {\bibinfo {title} {{Benchmarking Quantum Computers: The Five-Qubit Error Correcting Code}},\ }\href {https://doi.org/10.1103/PhysRevLett.86.5811} {\bibfield  {journal} {\bibinfo  {journal} {Physical Review Letters}\ }\textbf {\bibinfo {volume} {86}},\ \bibinfo {pages} {5811} (\bibinfo {year} {2001})}\BibitemShut {NoStop}%
\bibitem [{\citenamefont {Le~Régent}(2025)}]{Le2025AwesomeComputation}%
  \BibitemOpen
  \bibfield  {author} {\bibinfo {author} {\bibfnamefont {F.-M.}\ \bibnamefont {Le~Régent}},\ }\bibfield  {title} {\bibinfo {title} {{Awesome Quantum Computing Experiments: Benchmarking Experimental Progress Towards Fault-Tolerant Quantum Computation}},\ }\href {https://github.com/francois-} {\bibfield  {journal} {\bibinfo  {journal} {ArXiv}\ } (\bibinfo {year} {2025})},\ \Eprint {https://arxiv.org/abs/2507.03678v1} {arXiv:2507.03678v1 [quant-ph]} \BibitemShut {NoStop}%
\bibitem [{\citenamefont {Zobrist}\ \emph {et~al.}(2025)\citenamefont {Zobrist}, \citenamefont {Zhu}, \citenamefont {Zhang}, \citenamefont {Zalcman} \emph {et~al.}}]{Zobrist2025QuantumThreshold}%
  \BibitemOpen
  \bibfield  {author} {\bibinfo {author} {\bibfnamefont {N.}~\bibnamefont {Zobrist}}, \bibinfo {author} {\bibfnamefont {N.}~\bibnamefont {Zhu}}, \bibinfo {author} {\bibfnamefont {Y.}~\bibnamefont {Zhang}}, \bibinfo {author} {\bibfnamefont {A.}~\bibnamefont {Zalcman}}, \emph {et~al.},\ }\bibfield  {title} {\bibinfo {title} {{Quantum error correction below the surface code threshold}},\ }\href {https://doi.org/10.1038/S41586-024-08449-Y} {\bibfield  {journal} {\bibinfo  {journal} {Nature}\ }\textbf {\bibinfo {volume} {638}},\ \bibinfo {pages} {920} (\bibinfo {year} {2025})}\BibitemShut {NoStop}%
\bibitem [{\citenamefont {Tremba}\ \emph {et~al.}(2025)\citenamefont {Tremba}, \citenamefont {Hovland},\ and\ \citenamefont {Liu}}]{Tremba2025IsRuntimes}%
  \BibitemOpen
  \bibfield  {author} {\bibinfo {author} {\bibfnamefont {M.}~\bibnamefont {Tremba}}, \bibinfo {author} {\bibfnamefont {P.}~\bibnamefont {Hovland}},\ and\ \bibinfo {author} {\bibfnamefont {J.}~\bibnamefont {Liu}},\ }\bibfield  {title} {\bibinfo {title} {{Is Circuit Depth Accurate for Comparing Quantum Circuit Runtimes?}},\ }\href {https://arxiv.org/pdf/2505.16908v1} {\bibfield  {journal} {\bibinfo  {journal} {ArXiv}\ } (\bibinfo {year} {2025})},\ \Eprint {https://arxiv.org/abs/2505.16908v1} {arXiv:2505.16908v1 [quant-ph]} \BibitemShut {NoStop}%
\bibitem [{\citenamefont {Dasu}\ \emph {et~al.}(2025)\citenamefont {Dasu}, \citenamefont {Criger}, \citenamefont {Foltz}, \citenamefont {Gerber}, \citenamefont {Gilbreth}, \citenamefont {Gilmore}, \citenamefont {Holliman}, \citenamefont {Lysne}, \citenamefont {Milne}, \citenamefont {Okuno}, \citenamefont {Vittorini},\ and\ \citenamefont {Hayes}}]{Dasu2025Order-of-magnitudeCode}%
  \BibitemOpen
  \bibfield  {author} {\bibinfo {author} {\bibfnamefont {S.}~\bibnamefont {Dasu}}, \bibinfo {author} {\bibfnamefont {B.}~\bibnamefont {Criger}}, \bibinfo {author} {\bibfnamefont {C.}~\bibnamefont {Foltz}}, \bibinfo {author} {\bibfnamefont {J.~A.}\ \bibnamefont {Gerber}}, \bibinfo {author} {\bibfnamefont {C.~N.}\ \bibnamefont {Gilbreth}}, \bibinfo {author} {\bibfnamefont {K.}~\bibnamefont {Gilmore}}, \bibinfo {author} {\bibfnamefont {C.~A.}\ \bibnamefont {Holliman}}, \bibinfo {author} {\bibfnamefont {N.~K.}\ \bibnamefont {Lysne}}, \bibinfo {author} {\bibfnamefont {A.~R.}\ \bibnamefont {Milne}}, \bibinfo {author} {\bibfnamefont {D.}~\bibnamefont {Okuno}}, \bibinfo {author} {\bibfnamefont {G.}~\bibnamefont {Vittorini}},\ and\ \bibinfo {author} {\bibfnamefont {D.}~\bibnamefont {Hayes}},\ }\bibfield  {title} {\bibinfo {title} {{Order-of-magnitude extension of qubit lifetimes with a decoherence-free subspace quantum error correction code}},\ }\href@noop {} {\bibfield  {journal} {\bibinfo  {journal} {ArXiv}\ }
  (\bibinfo {year} {2025})},\ \Eprint {https://arxiv.org/abs/2503.22107v1} {arXiv:2503.22107v1 [quant-ph]} \BibitemShut {NoStop}%
\end{thebibliography}%
\newpage
\appendix
\section{State preparation subroutine evaluation}
\label{app: SP}
QCLib is a library that has collected and compared several state preparation subroutines \cite{Araujo2023QuantumLibrary}. We repeated their analysis for our Gaussian initial condition $\ket{p_0}$ and found that the low-rank subroutine of state preparation reduced the gate count significantly. 
The results are shown in \cref{fig:SP_comp}.  The low-rank subroutine scales more efficiently with $q$ than the naive state preparation subroutine, where $q = d\log_2 n_x$. One of the key improvements is that the low-rank subroutine does not have a $d$-dependence. This is due to the low-rank subroutine taking advantage of the sparseness of $\mathbf{p}_0$ via Schmidt decomposition. We can see this by the dip in the low-rank depth at 8 qubits, 8 qubits corresponds to $d=2$ and $n_x = 16$. The reduction achieved by the low-rank subroutine is catalogued in ~\cref{tab:state_prep}.

\begin{figure}
\floatsetup{valign=c, heightadjust=object}
\begin{floatrow}
\ffigbox{\scalebox{0.8}{\begin{tikzpicture}
\begin{axis}[
    xlabel={Qubits},
    ylabel={Depth},
    xmin=2, xmax=14,
    ymin=1, ymax=10000,
    ytick={1, 10, 100, 1000, 10000},
    yticklabels={1, 10, 100, 1000, 10000},
    xtick={2, 4, 6, 8, 10, 12, 14},
    legend pos= north west,
    ymajorgrids=true,
    ymode = log,
    log basis y={10},
    xmajorgrids=true,
    axis x line*=top,
]
    \addlegendimage{only marks, mark=x, color=blue}
    \addlegendentry{State Prep.}
    \addlegendimage{only marks, mark=x, color= red}
    \addlegendentry{Low-rank}
    \addlegendimage{only marks, mark=x, color=olive}
    \addlegendentry{Theoretical}
    \addlegendimage{only marks, mark=x, color = black}
    \addlegendentry{d = 1}
    \addlegendimage{only marks, mark=o, color=black}
    \addlegendentry{d = 2} 
    \addlegendimage{only marks, mark=triangle, color=black}
    \addlegendentry{d = 3} 
\addplot[
    color=blue,
    mark=x,
    only marks,
    mark size = 3pt,
    ]
    coordinates {
    (4,26)(5,57)(6,120)
    };
\addplot[
    color=blue,
    mark=o,
    only marks,
    mark size = 3pt,
    ]
    coordinates {
    (8,438)(10,2036)
    };
\addplot[
    color=blue,
    mark=triangle,
    only marks,
    mark size = 3pt,
    ]
    coordinates {
    (12,8114)
    };
\addplot[
    color=red,
    mark=x,
    only marks,
    mark size = 3pt,
    ]
    coordinates {
    (4,6)(5,22)(6,23)
    };
\addplot[
    color=red,
    mark=o,
    only marks,
    mark size = 3pt,
    ]
    coordinates {
    (8,5)(10,22)
    };
\addplot[
    color=red,
    mark=triangle,
    only marks,
    mark size = 3pt,
    ]
    coordinates {
    (12,219)
    };
\addplot[
    color=olive,
    mark=x,
    only marks,
    mark size = 3pt,
    ]
    coordinates {
    (4,4)(5,5)(6,6)
    };
\addplot[
    color=olive,
    mark=o,
    only marks,
    mark size = 3pt,
    ]
    coordinates {
    (8,8)(10,10)
    };
\addplot[
    color=olive,
    mark=triangle,
    only marks,
    mark size = 3pt,
    ]
    coordinates {
    (12,12)
    };
   
\end{axis}
\begin{axis}[
        xmin=1, xmax=7,
        xtick={1, 2, 3, 4, 5, 6, 7},
        xticklabels={4, 16, 64, 16, 32, 16},
        ymin=0, ymax=70,
        hide y axis,
        axis x line*=bottom,
        xlabel={$n_x$}
    ]
  \end{axis}

\end{tikzpicture}}}{\caption{A comparison of circuit depth of two state preparation subroutines in the Tket gate set.}\label{fig:SP_comp}}
\capbtabbox{\scalebox{0.85}{
    \begin{tabular}{c c c c c}
\hline
Dimension & Discretisation steps & \multicolumn{2}{c}{Circuit Depth} & Reduction \\
 ($d$) &  ($n_x$) & tket & Low-rank & (\%) \\
\hline 
1 & 16 & 26 & 6 & 77\% \\
2 & 16 & 503 & 6 & 98.8\% \\
3 & 16 & 8179 & 238 & 97.0\% \\ 
\hline
\end{tabular}}}{\caption{Comparison of circuit depths for tket and low-rank state preparation subroutines across different test scenarios.} \label{tab:state_prep}}
    \end{floatrow}
\end{figure}

\section{Hoeffding's inequality and measurement analysis}
\label{app: Hoeff}
In this appendix, we provide the proof of Lemma~\ref{lem: Hoeff}. We then discuss how this compares to the error scaling with the number of shots, $N$.
\setcounter{theorem}{2}
\begin{lemma}[(Restated). Hoeffding's inequality]
    Given a quantum algorithm that produces the state 
    \begin{equation}
        |\Tilde{\mathbf{p}}\rangle = \frac{1}{\sqrt{\sum_{(\mathbf{x},t)\in G}\Tilde{p}(\mathbf{x})^2}}\sum_{(\mathbf{x})\in G} \Tilde{p}(\mathbf{x})|\mathbf{x}\rangle
    \end{equation}
    where $\mathbf{x} \in [0, n_x-1]^d$. To extract the probability distribution $\Tilde{p}(\mathbf{x})$ such that $||\Bar{\Tilde{p}}(\mathbf{x}) - \Tilde{p}(\mathbf{x})||_{\infty}\leq \epsilon_q$ with error probability $\delta$, $N$ measurements must be performed for each $n_x^d$ position, where $N \leq \dfrac{-\log(\delta/2)}{2\epsilon_q^2}$ and $\delta$ are defined by Hoeffding's inequality as follows
    \begin{equation}
        \mathrm{Prob}(||\Bar{\Tilde{\mathbf{p}}}_N - \Tilde{\mathbf{p}}||_{\infty}> \epsilon_q) \leq 2\mathrm{exp}\left(-2N\epsilon_q^2 \right)
    \end{equation}
    where $\Bar{\Tilde{p}}_N$ is the sampled mean over $N$ samples. 
\end{lemma}
\begin{proof}
Hoeffding's inequality states: Let $X_1, ...,X_N$ be independent random variables such that $a_i \leq X_i \leq b_i$. Consider the sum of these random variables
\begin{equation}
    S_N = X_1 + ... + X_N.
\end{equation}
Then Hoeffding's theorem states that for all $t>0$,
\begin{equation}
    \mathrm{Prob}(|S_N - E[S_N]|\geq t) \leq 2\mathrm{exp}\left(-\dfrac{2t^2}{\sum_{i=1}^{N}(b_i-a_i)^2}\right).
\end{equation}
When $a_i=0$ and $b_i=1$ for all $i$, we get the inequality
\begin{equation}
    \mathrm{Prob}(|S_N - E[S_N]|\geq t) \leq 2\mathrm{exp}\left(-2t^2/N\right)
\end{equation}
which is equivalent to
\begin{equation}
    \mathrm{Prob}(|(S_N - E[S_N])/N|\geq t) \leq 2\mathrm{exp}\left(-2t^2N\right).
\end{equation}
\end{proof}

As discussed in~\cref{sec: results} we can test this bound. From \cref{fig: n_x 1d}, we know that for $n_x = 32$  $\epsilon_c = 0.027$ and we desire an $\epsilon = 0.03$. Therefore, we require $\epsilon_q \leq 0.003$. Lemma~\ref{lem: Hoeff}, introduced an upper bound for the number of shots required, ensuring a success probability of $1-\delta$. 
Considering the scenario, $n_x = 32$ and $d=1$, we plot the results in \cref{fig:shots_vs_eps}. We observe a plateau at $N = 50,000$, this is equivalent to a success probability of $1-\delta = 0.81$.  If a higher confidence is desired, we can choose a success probability $1-\delta = 0.9$. Using Lemma~\ref{lem: Hoeff}, this results in an upper bound on the number of samples required of $N = 167,000$.

\begin{figure}
    \centering
    \scalebox{0.8}{
    \begin{tikzpicture}
\begin{axis}[
    xlabel={$N$},
    ylabel={$\epsilon_q$},
    xmin=0, xmax=1,
    ymin=0.0001, ymax=0.1,
    xtick={0, 0.2, 0.4, 0.6, 0.8, 1},
    xticklabels={0, 20000, 40000, 60000, 80000, 100000},
    ytick={0.0001, 0.001, 0.01, 0.1},
    yticklabels={0.0001, 0.001, 0.01, 0.1},
    ymajorgrids=true,
    xmajorgrids=true,
    ymode = log,
    log basis y={10}
]
    \addplot+ [only marks]
    table {%
    0.05 0.002327
    0.1 0.001601
    0.5 0.00103
    1 0.00101
    0.01 0.003438
    0.005 0.0885326
    0.001 0.02738
    };
\end{axis}
\end{tikzpicture}}
    \caption{The quantum error, $\epsilon_q$, scaling with number of shots, $N$, for $n_x = 32$, $d=1$.}
    \label{fig:shots_vs_eps}
\end{figure}
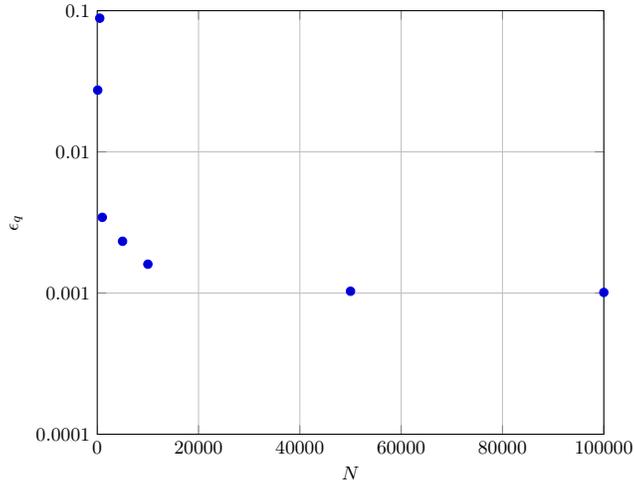

\section{Circuit output in two spatial dimensions}
\label{app: d2}
This appendix provides a validation of our quantum circuit for $d=2$. We compare the quantum and classical solution algorithms defined in~\cref{thm: QFT,thm: Classical_FFT} respectively. 
While the peaks are in the same locations, we can see that, in general, the quantum solution does not match the height of the classical one. This is due to the approximation in the \ac{FABLE} circuit, which, due to normalisation, reduces the overall height of the quantum result marginally. 

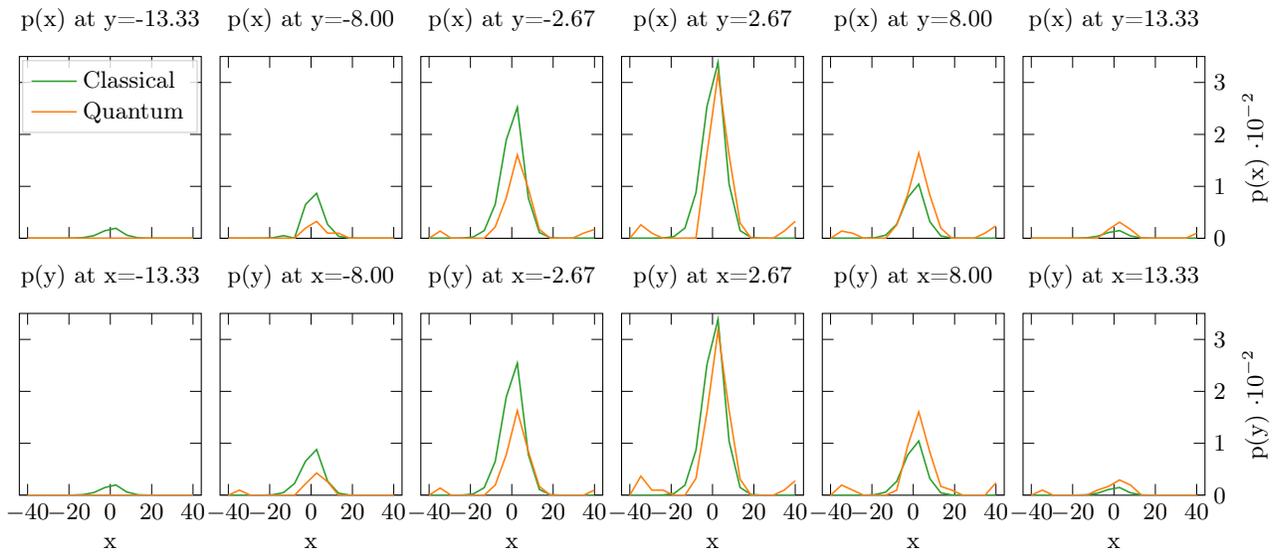
\begin{figure}
\noindent
\begin{tikzpicture}

\definecolor{darkgray176}{RGB}{176,176,176}
\definecolor{lightgray204}{RGB}{204,204,204}
\definecolor{crimson2143940}{RGB}{214,39,40}
\definecolor{darkgray176}{RGB}{176,176,176}
\definecolor{darkorange25512714}{RGB}{255,127,14}
\definecolor{forestgreen4416044}{RGB}{44,160,44}
\definecolor{lightgray204}{RGB}{204,204,204}
\definecolor{steelblue31119180}{RGB}{31,119,180}

\begin{groupplot}[group style={group size=6 by 2, 
            x descriptions at=edge bottom,
            y descriptions at=edge right,
            horizontal sep=0.25cm,
            vertical sep=1.0cm,}, 
            height = 4cm,
            width= 4cm,
            xlabel={x},
            xmin=-44, xmax=44,
            xtick style={color=black},
            ytick={0, 1, 2, 3},
            y grid style={darkgray176},
            ylabel={p(x) $\cdot 10^{-2}$},
            ymin=0, ymax=3.5,
            ytick style={color=black}]
\nextgroupplot[
title={p(x) at y=-13.33},legend cell align={left}, legend style={fill opacity=0.8, draw opacity=1, text opacity=1, draw=lightgray204}
]
\addplot [semithick, forestgreen4416044]
table {%
-40 0
-24 2.50339508056641e-04
-18.6666660308838 1.87158584594727e-03
-13.3333330154419 0.0113010406494141
-8 0.0502824783325195
-2.66666674613953 0.145196914672852
2.66666674613953 0.193405151367188
8 0.0595927238464355
13.3333330154419 8.47578048706055e-03
18.6666660308838 7.86781311035156e-04
29.3333339691162 0
40 0
};
\addlegendentry{Classical}
\addplot [semithick, darkorange25512714]
table {%
-40 0
40 0
};
\addlegendentry{Quantum}
\nextgroupplot[
title={p(x) at y=-8.00},
]
\addplot [semithick, forestgreen4416044]
table {%
-40 0
-24 1.09672546386719e-03
-18.6666660308838 8.40425491333008e-03
-13.3333330154419 0.050663948059082
-8 0.00225353240966797
-2.66666674613953 0.650703907012939
2.66666674613953 0.866734981536865
8 0.26707649230957
13.3333330154419 0.0379562377929688
18.6666660308838 3.50475311279297e-03
24 2.38418579101562e-04
40 0
};
\addplot [semithick, darkorange25512714]
table {%
-40 0
-8 0
-2.66666674613953 0.196301937103271
2.66666674613953 0.325524806976318
8 0.0981450080871582
13.3333330154419 0.0981450080871582
18.6666660308838 0
40 0
};

\nextgroupplot[
title={p(x) at y=-2.67},
]
\addplot [semithick, forestgreen4416044]
table {%
-40 0
-29.3333339691162 3.45706939697266e-04
-24 3.19480895996094e-03
-18.6666660308838 0.0244379043579102
-13.3333330154419 0.147378444671631
-8 0.655567646026611
-2.66666674613953 1.89297199249268
2.66666674613953 2.52143144607544
8 0.776946544647217
13.3333330154419 0.110423564910889
18.6666660308838 0.0101923942565918
24 6.91413879394531e-04
40 0
};
\addplot [semithick, darkorange25512714]
table {%
-40 0
-34.6666679382324 0.13880729675293
-29.3333339691162 0
-13.3333330154419 0
-8 0.219464302062988
-2.66666674613953 0.779044628143311
2.66666674613953 1.60378217697144
8 0.956642627716064
13.3333330154419 0.170004367828369
18.6666660308838 0
29.3333339691162 0
34.6666679382324 0.0981450080871582
40 0.170004367828369
};

\nextgroupplot[
title={p(x) at y=2.67},
]
\addplot [semithick, forestgreen4416044]
table {%
-40 0
-29.3333339691162 4.64916229248047e-04
-24 4.29153442382812e-03
-18.6666660308838 0.032806396484375
-13.3333330154419 0.197863578796387
-8 0.880157947540283
-2.66666674613953 2.54144668579102
2.66666674613953 3.38521003723145
8 1.04310512542725
13.3333330154419 0.148260593414307
18.6666660308838 0.0136733055114746
24 9.29832458496094e-04
34.6666679382324 0
40 0
};
\addplot [semithick, darkorange25512714]
table {%
-40 0
-34.6666679382324 0.259685516357422
-29.3333339691162 0.0981450080871582
-24 0
-13.3333330154419 0
-8 0.00427830219268799
-2.66666674613953 1.62763595581055
2.66666674613953 3.16979885101318
8 1.60378217697144
13.3333330154419 0.29444694519043
18.6666660308838 0
29.3333339691162 0
34.6666679382324 0.13880729675293
40 0.325524806976318
};

\nextgroupplot[
title={p(x) at y=8.00},
]
\addplot [semithick, forestgreen4416044]
table {%
-40 0
-24 1.32322311401367e-03
-18.6666660308838 0.0100970268249512
-13.3333330154419 0.0609159469604492
-8 0.270986557006836
-2.66666674613953 0.782489776611328
2.66666674613953 1.04227066040039
8 0.32116174697876
13.3333330154419 0.0456452369689941
18.6666660308838 4.20808792114258e-03
24 2.86102294921875e-04
40 0
};
\addplot [semithick, darkorange25512714]
table {%
-40 0
-34.6666679382324 0.13880729675293
-29.3333339691162 0.0981450080871582
-24 0
-13.3333330154419 0
-8 0.259685516357422
-2.66666674613953 0.844311714172363
2.66666674613953 1.63942575454712
8 0.832831859588623
13.3333330154419 0.196301937103271
18.6666660308838 0
29.3333339691162 0
34.6666679382324 0.0981450080871582
40 0.240421295166016
};

\nextgroupplot[
title={p(x) at y=13.33},
]
\addplot [semithick, forestgreen4416044]
table {%
-40 0
-24 1.9073486328125e-04
-18.6666660308838 1.43051147460938e-03
-13.3333330154419 8.65459442138672e-03
-8 0.0385046005249023
-2.66666674613953 0.111186504364014
2.66666674613953 0.148105621337891
8 0.0456333160400391
13.3333330154419 6.4849853515625e-03
18.6666660308838 5.96046447753906e-04
34.6666679382324 0
40 0
};
\addplot [semithick, darkorange25512714]
table {%
-40 0
-8 0
-2.66666674613953 0.170004367828369
2.66666674613953 0.310373306274414
8 0.170004367828369
13.3333330154419 0
34.6666679382324 0
40 0.0981450080871582
};


\nextgroupplot[
title={p(y) at x=-13.33},
]
\addplot [semithick, forestgreen4416044]
table {%
-40 0
-24 2.38418579101562e-04
-18.6666660308838 1.85966491699219e-03
-13.3333330154419 0.0113010406494141
-8 0.050663948059082
-2.66666674613953 0.147378444671631
2.66666674613953 0.197863578796387
8 0.0609159469604492
13.3333330154419 8.65459442138672e-03
18.6666660308838 7.98702239990234e-04
29.3333339691162 0
40 0
};
\addplot [semithick, darkorange25512714]
table {%
-40 0
40 0
};

\nextgroupplot[
title={p(y) at x=-8.00},
]
\addplot [semithick, forestgreen4416044]
table {%
-40 0
-24 1.07288360595703e-03
-18.6666660308838 8.27312469482422e-03
-13.3333330154419 0.0502824783325195
-8 0.225353240966797
-2.66666674613953 0.655567646026611
2.66666674613953 0.880157947540283
8 0.270986557006836
13.3333330154419 0.0385046005249023
18.6666660308838 3.55243682861328e-03
24 2.38418579101562e-04
40 0
};
\addplot [semithick, darkorange25512714]
table {%
-40 0
-34.6666679382324 0.0981450080871582
-29.3333339691162 0
-8 0
-2.66666674613953 0.219464302062988
2.66666674613953 0.427830219268799
8 0.259685516357422
13.3333330154419 0
40 0
};

\nextgroupplot[
title={p(y) at x=-2.67},
]
\addplot [semithick, forestgreen4416044]
table {%
-40 0
-29.3333339691162 3.33786010742188e-04
-24 3.09944152832031e-03
-18.6666660308838 0.0239014625549316
-13.3333330154419 0.145196914672852
-8 0.650703907012939
-2.66666674613953 1.89297199249268
2.66666674613953 2.54144668579102
8 0.782489776611328
13.3333330154419 0.111186504364014
18.6666660308838 0.0102519989013672
24 7.03334808349609e-04
40 0
};
\addplot [semithick, darkorange25512714]
table {%
-40 0
-34.6666679382324 0.13880729675293
-29.3333339691162 0
-13.3333330154419 0
-8 0.196301937103271
-2.66666674613953 0.779044628143311
2.66666674613953 1.62763595581055
8 0.844311714172363
13.3333330154419 0.170004367828369
18.6666660308838 0
34.6666679382324 0
40 0.0981450080871582
};

\nextgroupplot[
title={p(y) at x=2.67},
]
\addplot [semithick, forestgreen4416044]
table {%
-40 0
-29.3333339691162 4.41074371337891e-04
-24 4.13656234741211e-03
-18.6666660308838 0.0318288803100586
-13.3333330154419 0.193405151367188
-8 0.866734981536865
-2.66666674613953 2.52143144607544
2.66666674613953 3.38521003723145
8 1.04227066040039
13.3333330154419 0.148105621337891
18.6666660308838 0.0136613845825195
24 9.29832458496094e-04
34.6666679382324 0
40 0
};
\addplot [semithick, darkorange25512714]
table {%
-40 0
-34.6666679382324 0.367248058319092
-29.3333339691162 0.0981450080871582
-24 0.0981450080871582
-18.6666660308838 0
-13.3333330154419 0
-8 0.325524806976318
-2.66666674613953 1.60378217697144
2.66666674613953 3.16979885101318
8 1.63942575454712
13.3333330154419 0.310373306274414
18.6666660308838 0
29.3333339691162 0
34.6666679382324 0.0981450080871582
40 0.277614593505859
};

\nextgroupplot[
title={p(y) at x=8.00},
]
\addplot [semithick, forestgreen4416044]
table {%
-40 0
-24 1.27553939819336e-03
-18.6666660308838 9.8109245300293e-03
-13.3333330154419 0.0595927238464355
-8 0.26707649230957
-2.66666674613953 0.776946544647217
2.66666674613953 1.04310512542725
8 0.32116174697876
13.3333330154419 0.0456333160400391
18.6666660308838 4.20808792114258e-03
24 2.86102294921875e-04
40 0
};
\addplot [semithick, darkorange25512714]
table {%
-40 0
-34.6666679382324 0.196301937103271
-24 0
-13.3333330154419 0
-8 0.0981450080871582
-2.66666674613953 0.956642627716064
2.66666674613953 1.60378217697144
8 0.832831859588623
13.3333330154419 0.170004367828369
18.6666660308838 0.0981450080871582
24 0
34.6666679382324 0
40 0.240421295166016
};

\nextgroupplot[
title={p(y) at x=13.33},
ylabel={p(y) $\cdot 10^{-2}$}
]
\addplot [semithick, forestgreen4416044]
table {%
-40 0
-18.6666660308838 1.39474868774414e-03
-13.3333330154419 8.47578048706055e-03
-8 0.0379562377929688
-2.66666674613953 0.110423564910889
2.66666674613953 0.148260593414307
8 0.0456452369689941
13.3333330154419 6.4849853515625e-03
18.6666660308838 5.96046447753906e-04
34.6666679382324 0
40 0
};
\addplot [semithick, darkorange25512714]
table {%
-40 0
-34.6666679382324 0.0981450080871582
-29.3333339691162 0
-13.3333330154419 0
-8 0.0981450080871582
-2.66666674613953 0.170004367828369
2.66666674613953 0.29444694519043
8 0.196301937103271
13.3333330154419 0
40 0
};
\end{groupplot}

\end{tikzpicture}
\caption{The simulation results for the quantum and classical circuit for $d=2$, $n_x = 16$. Comparison in the $\infty$-norm gives $\epsilon_q = 0.01$.}
\label{fig: 2d_slices}
\end{figure}

\end{document}